%% file: main.tex
\documentclass[11pt]{article}
\input{header}

\date{}
\title{\Large \bf How to hide a clique?}
\author{
    {\rm Uriel Feige}\\
    uriel.feige@weizmann.ac.il \\
	Weizmann Institute of Science\\
	Rehovot, Israel
	\and
	{\rm Vadim Grinberg}\\
	vgm@ttic.edu\\
	Toyota Technological Institute at Chicago\footnote{Part of the work was done while the author was a visiting student in the Department of Computer Science and Applied Mathematics, Weizmann Institute of Science, and a full-time undergraduate student in the Faculty of Computer Science, Higher School of Economics, Moscow, Russia.}\\
	Chicago, IL, USA
}

\begin{document}

\maketitle

\begin{abstract}
In the well known planted clique problem, a clique (or alternatively, an independent set) of size $k$ is planted at random in an Erdos-Renyi random $G(n, p)$ graph, and the goal is to design an algorithm that finds the maximum clique (or independent set) in the resulting graph. We introduce a variation on this problem, where instead of planting the clique at random, the clique is planted by an adversary who attempts to make it difficult to find the maximum clique in the resulting graph. We show that for the standard setting of the parameters of the problem, namely, a clique of size $k = \sqrt{n}$ planted in a random $G(n, \frac{1}{2})$ graph, the known polynomial time algorithms can be extended (in a non-trivial way) to work also in the adversarial setting. In contrast, we show that for other natural settings of the parameters, such as planting an independent set of size $k=\frac{n}{2}$ in a $G(n, p)$ graph with $p = n^{-\frac{1}{2}}$, there is no polynomial time algorithm that finds an independent set of size $k$, unless NP has randomized polynomial time algorithms.
\end{abstract}

\input{introduction}

\subsection*{Acknowledgements}

The work of Uriel Feige is supported in part by the Israel Science Foundation (grant
No. 1388/16).
We are very grateful to Danila Kutenin for suggesting using the Bernstein inequality in \cref{varbound} to significantly simplify the proof.

\bibliographystyle{alpha}
\bibliography{main}

\appendix

\input{bounding}

\input{mainalgo}

\input{hardness_new}

\input{probbound}

\end{document}

%% file: header.tex
\usepackage[a4paper, margin=1in]{geometry}
\usepackage[utf8]{inputenc}         
\usepackage[english]{babel}
\usepackage{lmodern}
\usepackage{graphicx}

\usepackage[colorinlistoftodos]{todonotes}

\usepackage[colorlinks=true]{hyperref}

\usepackage{amsthm}
\usepackage{amsmath}
\usepackage{amssymb}
\usepackage{amsfonts}
\usepackage{mathrsfs}
\usepackage{mathtools}
\usepackage{verbatim}
\usepackage{footnote}
\usepackage{lineno}
\usepackage{caption}
\usepackage{tikzsymbols}
\usepackage[capitalize, nameinlink]{cleveref}

\usepackage{icomma}
\usepackage{enumitem}
\usepackage{array}
\usepackage{multirow}
\usepackage{setspace}
\usepackage{euscript}
\usepackage{indentfirst}
\usepackage{epigraph}
\usepackage{fancybox, fancyhdr}
\usepackage{titlesec}
\usepackage{dsfont}
\usepackage{csquotes}

\newtheorem{definition}{Definition}[section]

\newtheorem{theorem}{Theorem}[section]

\newtheorem{proposition}{Proposition}[section]

\newtheorem{corollary}{Corollary}[section]
\newtheorem{lemma}{Lemma}[section]

\newtheorem{claim}{Claim}[section]
\newtheorem{example}{Example}

\crefname{lemma}{\textbf{Lemma}}{}
\crefname{theorem}{\textbf{Theorem}}{}
\crefname{corollary}{\textbf{Corollary}}{}
\crefname{observation}{\textbf{Observation}}{}
\crefname{proposition}{\textbf{Proposition}}{}

\newcommand{\R}{\mathbb{R}}

\usepackage{xparse}

\DeclarePairedDelimiterX\brackets[1]{(}{)}{
	
	#1
}
\DeclarePairedDelimiterX\squarebrackets[1]{[}{]}{
	
	#1
}
\DeclarePairedDelimiterX\figbrackets[1]{\{}{\}}{
	
	#1
}

\let\P\undefined

\NewDocumentCommand{\P}{o g}{
	\IfNoValueTF{#1}{
		\IfNoValueTF{#2}{
			\mathop{\mathds{P}}
		}{
			\mathop{\mathds{P}}\squarebrackets*{#2}
		}
	}{
		\IfNoValueTF{#2}{
			\mathop{\mathds{P}}_{#1}
		}{
			\mathop{\mathds{P}}_{#1}\squarebrackets*{#2}
		}
	}
}

\NewDocumentCommand{\E}{o g}{
	\IfNoValueTF{#1}{
		\IfNoValueTF{#2}{
			\mathop{\mathds{E}}
		}{
			\mathop{\mathds{E}}\squarebrackets*{#2}
		}
	}{
		\IfNoValueTF{#2}{
			\mathop{\mathds{E}}_{#1}
		}{
			\mathop{\mathds{E}}_{#1}\squarebrackets*{#2}
		}
	}
}

\NewDocumentCommand{\D}{o g}{
	\IfNoValueTF{#1}{
		\IfNoValueTF{#2}{
			\mathop{\mathds{D}}
		}{
			\mathop{\mathds{D}}\squarebrackets*{#2}
		}
	}{
		\IfNoValueTF{#2}{
			\mathop{\mathds{D}}_{#1}
		}{
			\mathop{\mathds{D}}_{#1}\squarebrackets*{#2}
		}
	}
}

\NewDocumentCommand{\cov}{o g g}{
	\IfNoValueTF{#1}{
		\IfNoValueTF{#2}{
			\mathop{\mathrm{cov}}
		}{
			\IfNoValueTF{#3}{
				\mathop{\mathrm{cov}}
			}{
				\mathop{\mathrm{cov}}\brackets*{#2, #3}
			}
		}
	}{
		\IfNoValueTF{#2}{
			\mathop{\mathrm{cov}}_{#1}
		}{
			\IfNoValueTF{#3}{
				\mathop{\mathrm{cov}}_{#1}
			}{
				\mathop{\mathrm{cov}}_{#1}\brackets*{#2, #3}
			}
		}
	}
}

\NewDocumentCommand{\I}{g}{
	\IfNoValueTF{#1}{
		\mathop{\mathds{I}}
	}{
		\mathop{\mathds{I}}\figbrackets*{#1}
	}
}

\NewDocumentCommand{\V}{o g}{
	\IfNoValueTF{#1}{
		\IfNoValueTF{#2}{
			\mathop{\mathds{V}}
		}{
			\mathop{\mathds{V}}\squarebrackets*{#2}
		}
	}{
		\IfNoValueTF{#2}{
			\mathop{\mathds{V}}_{#1}
		}{
			\mathop{\mathds{V}}_{#1}\squarebrackets*{#2}
		}
	}
}

\newcommand{\eps}{\varepsilon}

\newcommand{\Bin}{\mathrm{Bin}}

\newcommand{\tr}{\mathrm{tr}}

\newcommand{\Rr}{\mathcal{R}}

\newcommand{\wR}{\widetilde{R}}

\newcommand{\oG}{\bar{G}}

\newcommand{\wG}{\tilde{G}}

%% file: introduction.tex
\section{Introduction}\label{sec:intro}

The planted clique problem, also referred to as hidden clique, is a problem of central importance in the design of algorithms. We introduce a variation of this problem where instead of planting the clique at random, an adversary plants the clique. Our main results are that in certain regimes of the parameters of the problem, the known polynomial time algorithms can be extended to work also in the adversarial settings, whereas for other regimes, the adversarial planting version becomes NP-hard. We find the results interesting for three reasons. One is that they concern an extensively studied problem (planted clique), but from a new direction, and we find that the results lead to a better understanding of what aspects of the planted clique problem are made use of by the known algorithms. Another is that extending the known algorithms (based on semidefinite programming) to the adversarial planted setting involves some new techniques regarding how semidefinite programming can be used and analysed. Finally, the NP-hardness results are interesting as they are proven in a semi-random model in which most of the input instance is random, and the adversary controls only a relatively small aspect of the input instance. One may hope that this brings us closer to proving NP-hardness results for purely random models, a task whose achievement would be a breakthrough in complexity theory.

\subsection{The random planted clique model}
\label{sec:random}

Our starting point is the Erdos-Renyi $G(n,  p)$ random graph model, which generates graphs on $n$ vertices, and every two vertices are connected by an edge independently with probability $p$. We start our discussion with the special case in which $p = \frac{1}{2}$, and other values of $p$ will be considered later. Given a graph $G$, let $\omega(G)$ denote the size of the maximum clique in $G$, and let $\alpha(G)$ denote the size of the maximum independent set. Given a distribution $D$ over graphs, we use the notation $G \sim D$ for denoting a graph sampled at random according to $D$. The (edge) complement of a graph $G \sim G(n,   \frac{1}{2})$ is by itself a graph sampled from $G(n, \frac{1}{2})$, and the complement of a clique is an independent set, and hence the discussion concerning cliques in $G(n, \frac{1}{2})$ extends without change to independent sets (and vice versa).

It is well known (proved by computing the expectation and variance of the number of cliques of the appropriate size) that for $G \sim G(n, \frac{1}{2})$, w.h.p. $\omega(G) \simeq 2\log n$ (the logarithm is in base~2). However, there is no known polynomial time algorithm that can find cliques of size $2\log n$ in such graphs. A polynomial time greedy algorithm can find a clique of size $(1 + o(1))\log n$. The existence of $\rho > 1$ for which polynomial time algorithms can find cliques of size $\rho \log n$ is a longstanding open problem.

In the classical planted clique problem, one starts with a graph $G' \sim G(n,  \frac{1}{2})$ and a parameter $k$. In $G'$ one chooses at random a set $K$ of $k$ vertices, and makes this set into a clique by inserting all missing edges between pairs of vertices with $K$. We refer to $K$ as the planted clique, and say that the resulting graph $G$ is distributed according to $G(n, \frac{1}{2}, k)$. Given $G \sim G(n, \frac{1}{2}, k)$, the algorithmic goal can be one of the following three: find $K$, find a clique of maximum size, or find any clique of size at least $k$. It is not difficult to show that when $k$ is sufficiently large (say, $k > 3\log n$), then with high probability $K$ is the unique maximum size clique in $G \sim G(n, \frac{1}{2}, k)$, and hence all three goals coincide. Hence in the planted clique problem, the goal is simply to design polynomial time algorithms that (with high probability over the choice of $G \sim G(n, \frac{1}{2}, k)$) find the planted clique $K$. The question is how large should $k$ be (as a function of $n$) so as to make this task feasible. 

For some sufficiently large constant $c > 0$ (throughout, we use $c$ to denote a sufficiently large constant), if $k > c\sqrt{n \log n}$, with high probability the the vertices of $K$ are simply the $k$ vertices of highest degree in $G$ (see~\cite{Kucera95}), and hence $K$ can easily be recovered. Alon, Krivelevich and Sudakov~\cite{AKS98} managed to shave the $\sqrt{\log n}$ factor, designing a spectral algorithm that recovers $K$ when $k > c\sqrt{n}$. They also showed that $c$ can be made an arbitrarily small constant, by increased the running time by a factor of $n^{O(\log(\frac{1}{c}))}$ (this is done by ``guessing" a set $K'$ of $O(\log(\frac{1}{c}))$ vertices of $K$, and finding the maximum clique in the subgraph induced on their common neighbors). Subsequently, additional algorithms were developed that find the planted clique when $k > c\sqrt{n}$. They include algorithms based on the Lovasz theta function, which is a form of semi-definite programming~\cite{FK00}, algorithms based on a ``reverse-greedy" principle~\cite{FR10, DGP14}, and message passing algorithms~\cite{DM15}. There have been many attempts to find polynomial time algorithms that succeed when $k = o(\sqrt{n})$, but so far all of them failed (see for example~\cite{Jerrum92, FK03, MPW15}). It is a major open problem whether there is any such polynomial time algorithm.

Planted clique when $p \not= \frac{1}{2}$ was not studied as extensively, but it is quite well understood how results from the $G(n, \frac{1}{2}, k)$ model transfer to the $G(n, p, k)$ model. For $p$ much smaller that $\frac{1}{2}$, say $p = n^{\delta-1}$ for some $0 < \delta < 1$ (hence average degree $n^{\delta}$), the problem changes completely. Even without planting, with high probability over the choice of $G \sim G(n, p)$ (with $p = n^{\delta-1}$) we have that $\omega(G) = O(\frac{1}{1 - \delta})$, and the maximum clique can be found in polynomial time. This also extends to finding maximum cliques in the planted setting, regardless of the value of $k$. (We are not aware of such results being previously published, but they are not difficult. See \cref{sec:enumeration}.) For $p > \frac{1}{2}$, it is more convenient to instead look at the equivalent problem in which $p < \frac{1}{2}$, but with the goal of finding a planted independent set instead of a planted clique. We refer to this model as $\bar{G}(n, p, k)$. For $G \sim G(n, p)$ (with $p = n^{\delta - 1}$) we have that with high probability $\alpha(G) = \Theta(n^{1 - \delta}\log n)$. For $G \sim \bar{G}(n, p, k)$ the known algorithms extend to finding planted independent sets of size $k = cn^{1 - \frac{\delta}{2}}$ in polynomial time. We remark that the approach of~\cite{AKS98} of making $c$ arbitrarily small does not work for such sparse graphs.

\subsection{The adversarial planted clique model}

In this paper we introduce a variation on the planted clique model (and planted independent set model) that we refer to as the adversarial planted clique model. As in the random planted clique model, we start with a graph $G' \sim G(n, p)$ and a parameter $k$. However, now a computationally unbounded adversary may inspect $G'$, select within it a subset $K$ of $k$ vertices of its choice, and make this set into a clique by inserting all missing edges between pairs of vertices with $K$. We refer to this model as $AG(n, p, k)$ (and the corresponding model for planted independent set as $A\bar{G}(n, p, k)$). As shorthand notation shall use $G \sim AG(n, p, k)$ to denote a graph generated by this process. Let us clarify that $AG(n, p, k)$ is not a distribution over graphs, but rather a family of distributions, where each adversarial strategy (where a strategy of an adversary is a mapping from $G'$ to a choice of $K$) gives rise to a different distribution. 

In the adversarial planted model, it is no longer true that the planted clique is the one of maximum size in the resulting graph $G$. Moreover, finding $K$ itself may be information theoretically impossible, as $K$ might be statistically indistinguishable from some other clique of size $k$ (that differs from $K$ by a small number of vertices). The three goals, that of finding $K$, finding a clique of maximum size, or finding any clique of size at least $k$, are no longer equivalent. Consequently, for our algorithmic results we shall aim at the more demanding goal of finding a clique of maximum size, whereas for our hardness results, we shall want them to hold even for the less demanding goal of finding an arbitrary clique of size $k$.

\subsection{Our results}

Our results cover a wide range of values of $0 < p < 1$, where $p$ may be a function of $n$. For simplicity of the presentation and to convey the main insights of our results, we present here the results for three representative regimes: $p = \frac{1}{2}$, $p = n^{\delta - 1}$ for $0 < \delta < 1$, and $p = 1 - n^{\delta - 1}$. For the latter regime, it will be more convenient to replace it by the equivalent problem of finding adversarially planted independent sets when $p = n^{\delta - 1}$.

Informally, our results show the following phenomenon. We consider only the case that $p \le \frac{1}{2}$, but consider both the planted clique and the planted independent set problems, and hence the results can be translated to $p > \frac{1}{2}$ as well. For clique, we show (\cref{thres1} and~\cref{thres2}) how to extend the algorithmic results known for the random planted clique setting to the adversarial planted clique setting. However, for independent set, we show that this is no longer possible. Specifically, when $p$ is sufficiently small, we prove (\cref{thres3}) that finding an independent set of size $k$ (any independent set, not necessarily the planted one) in the adversarial planted independent set setting is NP-hard. Moreover, the NP-hardness result holds even for large values of $k$ for which  finding a random planted independent set is trivial.

\begin{theorem}\label{thres1}
	For every fixed $\eps > 0$ and for every $k \geq \eps\sqrt{n}$, there is an (explicitly described) algorithm running in time $n^{O(\log(\frac{1}{\eps}))}$ which almost surely finds the maximum clique in a graph $G \sim AG(n, \frac{1}{2}, k)$. The statement holds for every adversarial planting strategy (choice of $k$ vertices as a function of $G' \sim G(n, \frac{1}{2})$), and the probability of success is taken over the choice of  $G' \sim G(n, \frac{1}{2})$.
\end{theorem}

\begin{theorem}\label{thres2}
Let $p = n^{\delta-1}$ for $0 < \delta < 1$. Then for every $k$, there is an (explicitly described) algorithm running in time $n^{O(\frac{1}{1-\delta})}$ which almost surely finds the maximum clique in a graph $G \sim AG(n, p, k)$. The statement holds for every adversarial planting strategy, and the probability of success is taken over the choice of  $G' \sim G(n, p)$.
\end{theorem}

\begin{theorem}\label{thres3}
For $p = n^{\delta-1}$ with $0 < \delta < 1$, $0 < \gamma < 1$, and $cn^{1 - \delta}\log n \le k \le \frac{2}{3}n$ (where $c$ is a sufficiently large constant, and the constant $\frac{2}{3}$ was chosen for concreteness -- any other constant smaller than~1 will work as well) the following holds.
There is no polynomial time algorithm that has probability at least $\gamma$ of finding an independent set of size $k$ in $G \sim A\bar{G}(n, p, k)$, unless NP has randomized polynomial time algorithms (NP=RP). (The algorithm is required to succeed against every adversarial planting strategy, and the probability of success is taken over the choice of  $G' \sim G(n, p)$.)
\end{theorem}

\subsection{Related work}

Some related work was already mentioned in \cref{sec:random}. 

Our algorithm for \cref{thres1} is based on an adaptation of the algorithm of~\cite{FK00} that applied to the random planted clique setting. In turn, that algorithm is based on the theta function of Lovasz~\cite{Lovasz}.

A work that is closely related to ours and served as an inspiration both to the model that we study, and to the techniques that are used in the proof of the NP-hardness result (\cref{thres3}) is the work of David and Feige~\cite{DF16} on adversarially planted 3-colorings. That work uncovers a phenomenon similar to the one displayed in the current work. Specifically, for the problem of 3-coloring (rather than clique or independent set) it shows that for certain values of $p$, algorithms that work in the random planted setting can be extended to the adversarial planted setting, and for other values of $p$, finding a 3-coloring in the adversarial planted setting becomes NP-hard. However, there are large gaps left open in the picture that emerges from the work of~\cite{DF16}. For large ranges of the values of $p$, specifically, $n^{-1/2} < p < n^{-1/3}$ and $p < n^{-2/3}$, there are neither algorithmic results nor hardness results in the work of~\cite{DF16}. Unfortunately, the most interesting values of $p$ for the 3-coloring problem, which are $p \le \frac{c\log n}{n}$, lie within these gaps, and hence the results of~\cite{DF16} do not apply to them. Our work addresses a different problem (planted clique instead of planted 3-coloring), and for our problem, our analysis leaves almost no such gaps. We are able to determine for which values of $p$ the problem is polynomial time solvable, and for which values it is NP-hard. See \cref{sec:discussion} for more details.

Our model is an example of a {\em semi-random} model, in which part of the input is determined at random and part is determined by an adversary. There are many other semi-random models, both for the clique problem and for other problems. Describing all these models is beyond the scope of this paper, and the interested reader is referred to~\cite{Fei20} and references therein for additional information. 


\section{Overview of the proofs}

In this section we provide an overview of the proofs for our three main theorems. Further details, as well as extensions to the results, appear in the appendix.

The term {\em almost surely} denotes a probability that tends to~1 as $n$ grows. The term {\em extremely high probability} denotes a probability of the form $1 - e^{-n^r}$ for some $r > 0$. 
By $\exp(x)$ for some expression $x$ we mean $e^x$.

\subsection{Finding cliques using the theta function}
\label{sec:algorithm}

In this section we provide an overview of the proof of \cref{thres1}. Our algorithm is an adaptation of the algorithm of~\cite{FK00} that finds the maximum clique in the random planted model. We shall first review that algorithm, then describe why it does not apply in our setting in which an adversary plants the clique, and finally explain how we modify that algorithm and its analysis so as to apply it in the adversarial planted setting. 

The key ingredient in the algorithm of~\cite{FK00} is the theta function of Lovasz, denoted by $\vartheta$. Given a graph $G$, $\vartheta(G)$ can be computed in polynomial time (up to arbitrary precision, using semidefinite programming (SDP)), and satisfies $\vartheta(G) \ge \alpha(G)$. As we are interested here in cliques and not in independent sets, we shall consider $\bar{G}$, the edge complement of $G$, and then $\vartheta(\bar{G}) \ge \omega(G)$. The theta function has several equivalent definitions, and the one that we shall use here (referred to as $\vartheta_4$ in~\cite{Lovasz}) is the following.

Given a graph $G = G(V, E)$, a collection of unit vectors $s_i \in \R^n$ (one vector for every vertex $i \in V$) is an {\em orthonormal representation} of $G$, if $s_i$ and $s_j$ are orthogonal ($s_i \cdot s_j = 0$) whenever $(i, j) \in E$. The theta function is the maximum value of the following expression, where maximization is over all orthonormal representations $\{s_i\}$ of $G$ and over all unit vectors $h$ ($h$ is referred to as the {\em handle}):

\begin{equation}\label{eq:thetaPrimal}
       \vartheta(G) = \max_{h,\{s_i\}} \sum_{i\in V} (h \cdot s_i)^2
\end{equation}

The optimal orthonormal representation and the associated handle that maximize the above formulation for $\vartheta$ can be found (up to arbitrary precision) in polynomial time by formulating the problem as an SDP (details omitted). Observe that for any independent set $S$ the following is a feasible solution for the SDP: choose $s_i = h$ for all $i \in S$, and choose all remaining vectors $s_j$ for $j \not\in S$ to be orthogonal to $h$ and to each other. Consequently, $\vartheta(G) \ge \alpha(G)$, as claimed.

The main content of the algorithm of~\cite{FK00} is summarized in the following theorem. We phrased it in a way that addresses cliques rather than independent sets, implicitly using $\alpha(\bar{G}) = \omega(G)$. We also remind the reader that in the random planted model, the planted clique $K$ is almost surely the unique maximum clique.

\begin{theorem}[Results of~\cite{FK00}]\label{thm:FK00}
Consider $G \sim G(n, \frac{1}{2}, k)$, a graph selected in the random planted clique model, with $k \ge c\sqrt{n}$ for some sufficiently large constant $c$.  Then with extremely high probability (over choice of $G$) it holds that $\vartheta(\bar{G}) = \omega(G)$.

Moreover, for every vertex $i$ that belongs to the planted clique $K$, the corresponding vector $s_i$ has inner product larger than $1 - \frac{1}{n}$ with the handle $h$, and for every other vertex, the corresponding inner product is at most $\frac{1}{n}$.
\end{theorem}  

Given \cref{thm:FK00}, the following algorithm finds the planted clique when $G \sim  G(n, \frac{1}{2}, k)$, and $k \ge c\sqrt{n}$ for some sufficiently large constant $c$. Solve the optimization problem (\ref{eq:thetaPrimal}) (on $\bar{G}$) to sufficiently high precision, and output all vertices whose corresponding inner product with $h$ is at least $\frac{1}{2}$.

The algorithm above does not apply to $G \sim AG(n, \frac{1}{2}, k)$, a graph selected in the adversarial planted clique model, for the simple reason that \cref{thm:FK00} is incorrect in that model. The following example illustrates what might go wrong, 

\begin{example}\label{ex:theta}
Consider a graph $G' \sim G(n, \frac{1}{2})$. In $G'$ first select a random vertex set $T$ of size slightly smaller than $\frac{1}{2}\log n$. Observe that the number of vertices in $G'$ that are in the common neighborhood of all vertices of $T$ is roughly $2^{-|T|}n > \sqrt{n}$. Plant a clique $K$ of size $k$ in the common neighborhood of $T$. In this construction, $K$ is no longer the largest clique in $G$. This is because $T$ (being a random graph) is expected to have a clique $K'$ of size $2\log |T| \simeq 2\log\log n$, and $K' \cup K$ forms a clique of size roughly $k + 2\log\log n$ in $G$. Moreover, as $T$ itself is a random graph with edge probability $\frac{1}{2}$, the value of the theta function on $T$ is roughly $\sqrt{|T|}$ (see~\cite{Juhasz}), and consequently one would expect the value of $\vartheta(\bar{G})$ to be roughly $k + \sqrt{\log n}$. 
\end{example}

Summarizing, it is not difficult to come up with strategies for planting cliques of size $k$ that result in the maximum clique having size strictly larger than $k$, and the value of $\vartheta(\bar{G})$ being even larger. Consequently, the solution of the optimization problem~(\ref{eq:thetaPrimal}) by itself is not expected to correspond to the maximum clique in $G$.

We now explain how we overcome the above difficulty. A relatively simple, yet important, observation is the following.

\begin{proposition}
\label{pro:maxIS}
Let $G \sim AG(n, p, k)$ with $p = 1/2$ and $k > \sqrt{n}$, and let $K'$ be the maximum clique in $G$ (which may differ from the planted clique $K$).
	Then with extremely high probability over the choice of $G' \sim G(n, \frac{1}{2})$, for every possible choice of $k$ vertices by the adversary, $K'$ contains at least $k - O(\log n)$ vertices from $K$, and at most $O(\log n)$ additional vertices.
\end{proposition}

\begin{proof}
Standard probabilistic arguments show that with extremely high probability, the largest clique in $G'$ (prior to planting a clique of size $k$) is of size at most $\frac{k}{2}$. When this holds, $K'$ contains at least $\frac{k}{2}$ vertices from $K$. Each of the remaining vertices of $K'$ needs to be connected to all vertices in $K' \cap K$. Consequently, with extremely high probability, $K'$ contains at most $2\log n$ vertices not from $K$. This is because a $G' \sim G(n, \frac{1}{2})$ graph, with extremely high probability, does not contain two sets of vertices $A$ and $B$, with $|A| = 2\log n$, $|B| = \Omega(\sqrt{n})$, such that all pairs of vertices in $A \times B$ induce edges in $G$.

As $|K'| \ge k$, we conclude that all but $O(\log n)$ vertices of $K$ must be members of $K'$.
\end{proof}

A key theorem that we prove is:

\begin{theorem}\label{thm:ourtheta}
	Let $G \sim AG(n, p, k)$ with $p = 1/2$ and $k = k(n) \geq 10\sqrt{n}$.
	Then $k \leq \vartheta(\bar{G}) \leq k + O(\log n)$ with extremely high probability over the choice of $G' \sim G(n, \frac{1}{2})$, for every possible choice of $k$ vertices by the adversary.
\end{theorem}

We now explain how \cref{thm:ourtheta} is proved. The bound $\vartheta(\bar{G}) \ge k$ was already explained above. Hence it remains to show that $\vartheta(\bar{G}) \leq k + O(\log n)$. In general, to bound $\vartheta(G)$ from above for a graph $G(V, E)$, one considers the following dual formulation of $\vartheta$, as a minimization problem.  
\begin{equation}\label{eq:thetaDual}
       \vartheta(G) = \min_M[\lambda_1(M)]
\end{equation}
Here $M$ ranges over all $n$ by $n$ symmetric matrices in which $M_{ij} = 1$ whenever $(i, j) \not\in E$, and $\lambda_1(M)$ denotes the largest eigenvalue of $M$. (Observe that if $G$ has an independent set $S$ of size $k$, then $M$ contains a $k$ by $k$ block of~1 entries. A Rayleigh quotient argument then implies that $\lambda_1(M) \ge k$, thus verifying the inequality $\vartheta(G) \ge \alpha(G)$.) To prove \cref{thm:ourtheta} we exhibit a matrix $M$ as above (for the graph $\bar{G}$) for which we prove that $\lambda_1(M) \le k + O(\log n)$.

We first review how a matrix $M$ was chosen by~\cite{FK00} in the proof of \cref{thm:FK00}. First, recall that we consider $\bar{G}$, and let $E$ be the set of edges of $\bar{G}$ (non-edges of $G$).
We need to associate values with the entries $M_{ij}$ for $(i, j)\in E$ (as other entries are~1). The matrix block corresponding to the planted clique $K$ (planted independent set in $\bar{G}$) is all~1 (by necessity). For every $(i, j)\in E$ where both vertices are not in $K$ one sets $M_{ij} = -1$. For every other pair $(i, j)\in E$ (say, $i \in K$ and $j \not\in K$) one sets $M_{i,j} = -\frac{k - d_{i,K}}{d_{i,K}}$, where $d_{i,K}$ is the number of neighbors that vertex $i$ has in the set $K$. In order to show that $\lambda_1(M) = k$, one first observes that the vector $x_K$ (with value~1 at entries that correspond to vertices of $K$, and value~0 elsewhere) is an eigenvector of $M$ with eigenvalue $k$. Then one proves that $\lambda_2(M)$, the second largest eigenvalue of $M$, has value smaller than $k$. This is done by decomposing $M$ into a sum of several matrices, bounding the second largest eigenvalue for one of these matrices, and the largest eigenvalue for the other matrices. By Weyl's inequality, the sum of these eigenvalues is an upper bound on $\lambda_2(M)$. This upper bound is not tight, but it does show that $\lambda_2(M) < k$. It follows that the eigenvalue $k$ associated with $x_K$ is indeed $\lambda_1(M)$. Further details are omitted.

We now explain how to choose a matrix $M$ so as to prove the bound $\vartheta(\bar{G}) \leq k + O(\log n)$ in \cref{thm:ourtheta}. Recall (see \cref{ex:theta}) that we might be in a situation in which $\vartheta(\bar{G}) > \alpha(\bar{G}) > k$ (with all inequalities being strict). In this case, let $K'$ denote the largest independent set in $\bar{G}$, and note that $K'$ is larger than $K$. In $M$, the matrix block corresponding to $K'$ is all~1. One may attempt to complete the construction of $M$ as described above for the random planting case, but replacing $K$ by $K'$ everywhere in that construction. If one does so, the vector $x_{K'}$ (with value~1 at entries that correspond to vertices of $K'$, and value~0 elsewhere) is an eigenvector of $M$ with eigenvalue $\alpha(\bar{G}) > k$. However, $M$ would necessarily have another eigenvector with a larger eigenvalue, because $\vartheta(\bar{G}) > \alpha(\bar{G})$. Hence we are still left with the problem of bounding $\lambda_1(M)$, rather than bounding $\lambda_2(M)$. Having failed to identify an eigenvector for $\lambda_1(M)$, we may still obtain an upper bound on $\lambda_1(M)$ by using approaches based on  Weyl's inequality (or other approaches). However, these upper bounds are not tight, and it seems difficult to limit the error that they introduce to be as small as $O(\log n)$, which is needed for proving the inequality $\lambda_1(M) \le k + O(\log n)$. 


For the above reason, we choose $M$ differently. For some constant $\frac{1}{2} < \rho < 1$, we extend the clique $K$ to a possibly larger clique $Q$, by adding to it every vertex that has  $\rho k$ neighbors in $K$. (In \cref{ex:theta}, the corresponding clique $Q$ will include all vertices of $K \cup T$. In contrast, if $K$ is planted at random and not adversarially, then we will simply have $Q = K$.) Importantly, we prove (see \cref{Gcloseneigh}) that if $G' \sim G(n, \frac{1}{2})$, then with high probability $|Q| < k + O(\log n)$ (for every possible choice of planting a clique of size $k$ by the adversary). 
For the resulting graph $G_Q$, we choose the corresponding matrix $M$ in the same way as it was chosen for the random planting case. 
Now we do manage to show that the eigenvector $x_Q$ (with eigenvalue $|Q|$) associated with this $M$ indeed has the largest eigenvalue. This part is highly technical, and significantly more difficult than the corresponding proof for the random planting case. 
The reason for the added level of difficulty is that, unlike the random planting case in which we are dealing with only one random graph, here the adversary can plant the clique in any one of ${n \choose k}$ locations, and our analysis needs to hold simultaneously for all ${n \choose k}$ graphs that may result from such plantings.  Further details can be found in \cref{sec:bounding}. 

Having established that $\vartheta(\bar{G}_Q) = |Q| \le k + O(\log n)$, we use monotonicity of the theta function to conclude that $\vartheta(\bar{G}) \le k + O(\log n)$. This concludes our overview for the proof of \cref{thm:ourtheta}.
 
Given \cref{thm:ourtheta}, let us now explain our algorithm for finding a maximum clique in $G \sim AG(n, \frac{1}{2}, k)$. 

Given a graph $G \sim AG(n, \frac{1}{2}, k)$, the first step in our algorithm is to solve the optimization problem~(\ref{eq:thetaPrimal}) on the complement graph $\bar{G}$. By \cref{thm:ourtheta}, we will have $\vartheta(\bar{G}) \le k + c\log n$ for some constant $c > 0$. Let $\{s_i\}$ denote the orthonormal representation found by our solution, and let $h$ be the corresponding handle.

The second step of our algorithm it to extract from $G$ a set of vertices that we shall refer to as $H$, that contains all those vertices $i$ for which $(h \cdot s_i)^2 \ge \frac{3}{4}$. 

\begin{lemma}
\label{lem:setH}
For $H$ as defined above, with extremely high probability, at least $k - O(\log n)$ vertices of $K$ are in $H$, and most $O(\log n)$ vertices not from $K$ are in $H$. 
\end{lemma}

\begin{proof}
Let $T$ denote the set of those vertices in $K$ for which $(h \cdot s_i)^2 < \frac{3}{4}$. Remove $T$ from $G$, thus obtaining the graph $G_T$. This graph can be thought of as a subgraph with $n - |T|$ vertices of the random graph $G' \sim G(n, \frac{1}{2})$, in which an adversary planted a clique of size $k - |T|$. We also have that $\vartheta(\bar{G}_T) \ge \vartheta(\bar{G}) - \sum_{i\in T} (h\cdot s_i)^2 \ge k  - \frac{3}{4}|T|$. If $|T|$ is large (larger than $c'\log n$ for some sufficiently large constant $c' > 0$), the gap of $\frac{|T|}{4}$ between the size of the planted clique and the value of the theta function contradicts \cref{thm:ourtheta} for the graph $G_T$. (Technical remark: this last argument uses the fact that \cref{thm:ourtheta} holds with extremely high probability, as we take a union bound over all choices of $T$.)

Having established that $T$ is small, let $R$ be the set of vertices not in $K$  for which $(h \cdot s_i)^2 \ge \frac{3}{4}$. We claim that every such vertex $i\in R$ is a neighbor of every vertex $j\in K \setminus T$. This is because in the orthogonal representation (for $\bar{G}$), if $i$ and $j$ are not neighbors we have that $s_i \cdot s_j = 0$, and then the fact that $s_i,s_j$ and $h$ are unit vectors implies that $(h\cdot s_i)^2 < 1 - (h\cdot s_j)^2 \le \frac{1}{4}$. Having this claim and using the fact that $|K \setminus T| > \sqrt{n}$, it follows that $|R| \le 2\log n$. This is because a $G' \sim G(n, \frac{1}{2})$ graph, with extremely high probability, does not contain two sets of vertices $A$ and $B$, with $|A| = 2\log n$, $|B| = \sqrt{n}$, such that all pairs of vertices in $A \times B$ induce edges in $G$.
\end{proof}


The third step of our algorithm constructs a set $F$ that contains all those vertices that have at least $\frac{3k}{4}$ neighbors in $H$. 

\begin{lemma}
\label{lem:setF}
With extremely high probability, the set $F$ described above contains the maximum clique in $G$, and at most $O(\log n)$ additional vertices.
\end{lemma} 

\begin{proof}
We may assume that $H$ satisfies the properties of \cref{lem:setH}. \cref{pro:maxIS} then implies that with extremely high probability, every vertex of the maximum clique in $G$ has at least $\frac{3k}{4}$ neighbors in $H$, and hence is contained in $F$. A probabilistic argument (similar to the end of the proof of \cref{lem:setH}) establishes that $F$ has at most $O(\log n)$ vertices not from $K$. As $K$ itself has at most $O(\log n)$ vertices not from the maximum clique (by \cref{pro:maxIS}), the total number of vertices in $F$ that are not members of the maximum clique is at most $O(\log n)$.
\end{proof}

Finally, in the last step of our algorithm we find a maximum clique in $F$, and this is a maximum clique in $G$. This last step can be performed in polynomial time by a standard algorithm (used for example to show that vertex cover is fixed parameter tractable). For every non-edge in the subgraph induced on $F$, at least one of its end-vertices needs to be removed. Try both possibilities in parallel, and recurse on each subgraph that remains. The recursion terminates when the graph is a clique. The shortest branch in the recursion gives the maximum clique. As only $O(\log n)$ vertices need to be removed in order to obtain a clique, the depth of the recursion is at most $O(\log n)$, and consequently the running time (which is exponential in the depth) is polynomial in $n$.

This completes our overview of our algorithm for finding a clique in $G \sim AG(n, \frac{1}{2}, k)$ when $k > c\sqrt{n}$ for a sufficiently large constant $c > 0$. To complete the proof of \cref{thres1} we need to also address the case that $k > \eps \sqrt{n}$ for arbitrarily small constant $\eps$. This we do (as in~\cite{AKS98}) by guessing $t \simeq 2\log \frac{c}{\epsilon}$ vertices from $K$ (there are $n^{t}$ possibilities to try, and we try all of them), and considering the subgraph of $G$ induced on their common neighbors. This subgraph corresponds to a subgraph of $G' \simeq G(n, \frac{1}{2})$ with roughy $n'\simeq  2^{-t} n$ vertices, and a planted clique of size $\eps\sqrt{n} - t \simeq c\sqrt{n'}$. Now on this new graph $G"$ we can invoke the algorithm based on the theta function. 
(Technical remark. The proof that $\vartheta(\bar{G}") \le k + O(\log n)$ uses the fact that \cref{thm:ourtheta} holds with extremely high probability. See more details in \cref{sec:mainalgo}.)

The many details that were omitted from the above overview of the proof of \cref{thres1} can be found in in the appendix.
Specifically, in \cref{sec:bounding} we present the proof of \cref{thm:ourtheta}, generalized to values of $p$ other than $1/2$, and $k \ge c\sqrt{np}$. (A technical lemma that is needed for this proof appears in \cref{sec:probbound}.) 
In \cref{sec:mainalgo} we present the proof of \cref{thres1}, first addressing the case that $c$ is sufficiently large, and  then extending the results to the case that $c$ can be arbitrarily small.

\subsection{Finding cliques by enumeration}
\label{sec:enumeration}

In this section we prove \cref{thres2}. 

Let $p = n^{\delta - 1}$ for $0 < \delta < 1$, and consider first $G' \sim G(n, p)$ (hence $G'$ has average degree roughly $n^{\delta}$). For every size $t \ge 1$, let $N_t$ denote the number of cliques of size $t$ in $G'$. The expectation (over choice of $G' \sim G(n, p)$) satisfies:
$$\E{N_t} = \binom{n}{t}p^{\binom{t}{2}} \le \frac{1}{t!}n^{\frac{\delta - 1}{2}t^2 + \frac{3-\delta}{2}t}$$
The exponent is maximized when $t = \frac{3 - \delta}{2(1 - \delta)}$. For the maximizing (not necessarily integer) $t$, the exponent equals $\frac{(3-\delta)^2}{8(1 - \delta)}$. We denote this last expression by $e_{\delta}$, and note that $e_{\delta} = O(\frac{1}{1 - \delta})$. 
The expected number of cliques of all sizes is then:
$$\sum_{t\ge 1} \E{N_t} \le n + \sum_{t\ge 2} \frac{1}{t!}n^{\frac{\delta - 1}{2}t^2 + \frac{3-\delta}{2}t} \le n^{e_{\delta}}$$
(The last inequality holds for sufficiently large $n$.)
By Markov's inequality, with probability at least $1 - \frac{1}{n}$, the actual number of cliques in $G'$ is at most $n^{e_{\delta} + 1}$. (Stronger concentration results can be used here, but are not needed for the proof of \cref{thres2}.)


Now, for arbitrary $1 \le k \le n$, let the adversary plant a clique $K$ of size $k$ in $G'$, thus creating the graph $G \sim G(n, p, k)$. As every subgraph of $K$ is a clique, the total number of cliques in $G$ is at least $2^k$, which might be exponential in $n$ (if $k$ is large). However, the number of maximal cliques in $G$ (a clique is maximal if it is not contained in any larger clique) is much smaller. Given a maximal clique $C$ in $G$, consider $C'$, the subgraph of $C$ not containing any vertex from $K$. $C'$ is a clique in $G'$ (which is nonempty, except for one special case of $C = K$). $C'$ uniquely determines $C$, as the remaining vertices in $C$ are precisely the set of common neighbors of $C'$ in $K$ (this is because the clique $C$ is maximal). Consequently, the number of maximal cliques in $G$ is not larger than the number of cliques in $G'$. 

As all maximal cliques in a graph can be enumerated in time linear in their number times some polynomial in $n$ (see e.g. \cite{MU04} and references therein), one can list all maximal cliques in $G$ in time $n^{e_{\beta} + O(1)}$ (this holds with probability at least $1 - \frac{1}{n}$, over the choice of $G'$, regardless of where the adversary plants clique $K$), and output the largest one.

This completes the proof of \cref{thres2}.

\subsection{Proving NP-hardness results}
\label{sec:hardness}

In this section we provide an overview of the proof of \cref{thres3}. Our proof is an adaptation to our setting of a proof technique developed in~\cite{DF16}.

Recall that we are considering a graph $G \sim A\bar{G}(n, p, k)$ (adversarial planted independent set) with $p = n^{\delta - 1}$ and $0 < \delta < 1$.  Let us first explain why the algorithm described in \cref{sec:algorithm} fails when $k = cn^{1 - \frac{\delta}{2}}$ (whereas if the independent set is planted at random, algorithms based on the theta function are known to succeed). The problem is that the bound in \cref{thm:ourtheta} is not true anymore, and instead one has the much weaker bound of $\vartheta(G) \le k + n^{1 - \delta}\log n$. Following the steps of the algorithm of \cref{sec:algorithm}, in the final step, we would need to remove a minimum vertex cover from $F$. However, now the upper bound on the size of this vertex cover is $O(n^{1 - \delta}\log n)$ rather than $O(\log n)$. Consequently, we do not know of a polynomial time algorithm that will do so. It may seem that we also do not know that no such algorithm exists. After all, $F$ is not an arbitrary worst case instance for vertex cover, but rather an instance derived from a random graph. However, our NP-hardness result shows that indeed this obstacle is insurmountable, unless NP has randomized polynomial time algorithms.  We remark that using an approximation algorithm for vertex cover in the last step of the algorithm of \cref{sec:algorithm} does allow one to find in $G$ an independent set of size $k - O(n^{1 - \delta}\log n) = (1 - o(1))k$, and the NP-hardness result applies only because we insist on finding an independent set of size at least $k$.

Let us proceed now with an overview of our NP hardness proof. We do so for the case that $k = \frac{n}{3}$ (for which we can easily find the maximum independent set if the planted independent set is random). Assume for the sake of contradiction that ALG is a polynomial time algorithm that with high probability over choice of $G' \sim G(n, p)$, for every planted independent set of size $k = \frac{n}{3}$, it finds in the resulting graph $G$ an independent set of size $k$.

We now introduce a class $\cal{H}$ of graphs that, in anticipation of the proofs that will follow, is required to have the following three properties. (Two of the properties are stated below in a qualitative manner, but they have precise quantitative requirements in the proofs that follow.)

\begin{enumerate}
    \item Solving maximum independent set on graphs from this class is NP-hard.
    \item Graphs in this class are very sparse.
    \item The number of vertices in each graph is small. 
\end{enumerate}

Given the above requirements, we choose  $0 < \eps < \min[\frac{\delta}{2}, 1 - \delta]$, and let $\cal{H}$ be the class of {\em balanced} graphs on $n^{\epsilon}$ vertices, and of average degree $2 + \delta$. (A graph $H$ is {\em balanced} if no subgraph of $H$ has average degree larger than the average degree of $H$.) Given a graph $H \in \cal{H}$ and a parameter $k'$, it is NP-hard to determine whether $H$ has an independent of size at least $k'$ or not (see \cref{nph}). We will reach a contradiction to the existence of ALG by showing how ALG could be used in order to find in $H$ an independent set of size $k'$, if one exists. For this, we use the following randomized algorithm ALGRAND.

\begin{enumerate}
    \item Generate a random graph $G' \sim G(n, p)$.
    \item Plant in $G'$ a random copy of $H$ (that is, pick $|H|$ random vertices in $G'$ and replace the subgraph induced on them by $H$). We refer to the resulting distribution as $G_H(n,p)$, and to the graph sampled from this distribution as $G_H$. 
    Observe that the number of vertices in $G_H$ that have a neighbor in $H$ is with high probability not larger than $|H|n^{\delta} \le \frac{n}{2}$.
    \item Within the non-neighbors of $H$, plant at random an independent set of size $k - k'$. We refer to the resulting distribution as $G_H(n,p,k)$, and to the graph sampled from this distribution as $\wG_H$. 
    Observe that with extremely high probability, $\alpha(\wG_H \setminus H) = k - k'$. Hence we may assume that this indeed holds. If furthermore $\alpha(H) \ge k'$, then $\alpha(\wG_H) \ge k$.
    \item Run ALG on $\wG_H$. We say that ALGRAND succeeds if ALG outputs an independent set $IS$ of size $k$. Observe that then at least $k'$ vertices of $H$ are in $IS$, and hence ALGRAND finds an independent set of size $k'$ in $H$. 
\end{enumerate}

If $H$ does not have an independent set of size $k'$, ALGRAND surely fails to output such an independent set. But if $H$ does have an independent set of size $k'$, why should ALGRAND succeed? This is because ALG (which is used in ALGRAND) is fooled to think that the graph $\wG_H$ generated by ALGRAND was generated from $A\bar{G}(n, p, k)$, and on such graphs ALG does find independent sets of size $k$. And why is ALG fooled? This is because the distribution of graphs generated by ALGRAND is statistically close to a distribution that can be created by the adversary in the $A\bar{G}(n, p, k)$ model. Specifically, consider the following distribution that we refer to as $A_HG(n, p, k)$.

\begin{enumerate}
    \item Generate $G' \sim G(n, p)$.
    \item The computationally unbounded  adversary finds within $G'$ all subsets of vertices of size $|H|$ such that the subgraph induced on them is $H$. (If there is no such subset, fail.) Choose one such copy of $H$ uniformly at random.
    \item As $H$ is assumed to have an independent set of size $k'$, plant an independent set $K$ of size $k$ as follows. $k'$ of the vertices of $K$ are vertices of an independent set in the selected copy of $H$. The remaining $k - k'$ vertices of $K$ are chosen at random among the vertices of $G'$ that have no neighbor at all in the copy of $H$. (Observe that we expect there to be at least roughly $n - |H|n^{\delta} \ge \frac{n}{2}$ such vertices, and with extremely high probability the actual number will be at least $\frac{n}{3} > k - k'$.)
\end{enumerate}

\begin{theorem}
\label{thm:closeStatistics}
The two distributions, $\wG_H \sim G_H(n, p, k)$ generated by ALGRAND and $G \sim A_HG(n, p, k)$ generated by the adversary, are statistically similar to each other.
\end{theorem}

The proof of \cref{thm:closeStatistics} appears in \cref{sec:proofThm2.3}. Here we explain the main ideas in the proof. A minimum requirement for the theorem to hold is that $G' \sim G(n, p)$ typically contains at least one copy of $H$ (otherwise $A_HG(n, p, k)$ fails to produce any output). But this by itself does not suffice. Intuitively, the condition we need is that $G'$ typically contains many copies of $H$. Then the fact that $G_H(n, p)$ of ALGRAND adds another copy of $H$ to $G'$ does not appear to make much of a difference to $G'$, because $G'$ anyway has many copies of $H$.  Hopefully, this will imply that $G' \sim G(n, p)$ and $G_H \sim G_H(n, p)$ come from two distributions that are statistically close. This intuition is basically correct, though another ingredient (a concentration result) is also needed. Specifically, we need the following lemma (stated informally).

\begin{lemma}
\label{lem:manyH}
For $G' \in G(n, p)$ (with $p$ and $H$ as above), the expected number of copies of $H$ in $G'$ is very high ($2^{n^{\eta}}$ for some $\eta > 0$ that depends on $\delta$ and $\epsilon$). Moreover, with high probability, the actual number of copies of $H$ in $G'$ is very close to its expectation.
\end{lemma}

The proof of \cref{lem:manyH} is based on known techniques (first and second moment methods). It uses in an essential way the fact that the graph $H$ is sparse (average degree barely above~2) and does not have many vertices (these properties hold by definition of the class $\cal{H}$).  See more details in \cref{sec:proofLem2.3}. Armed with \cref{lem:manyH}, we then prove the following Lemma. 

\begin{lemma}
\label{lem:closeStatistics}
The two distributions $G(n, p)$ and $G_H(n, p)$ are statistically similar to each other.
\end{lemma}

\cref{lem:closeStatistics} is proved by considering graphs $G' \sim G(n, p)$ that do contain a copy of $H$ (\cref{lem:manyH} establishes that this is a typical case), and comparing for each such graph the probability of it being generated by $G_H(n, p)$ with the probability of it being generated by $G(n, p)$. Conveniently, the ratio between these probabilities is the same as the ratio between the actual number of copies of $H$ in the given graph $G'$, and the expected number of copies of $H$ in a random $G' \sim G(n, p)$. By \cref{lem:manyH}, for most graphs, this ratio is close to~1. 
For more details, see \cref{sec:proofThm2.3}. 

\cref{thm:closeStatistics} follows quite easily from \cref{lem:closeStatistics}. Consequently ALG's performance on the distributions $G_H(n, p, k)$ and $A_HG(n,  p, k)$  is similar. By our assumption, ALG finds (with high probability) an independent set of size $k$ in $G \sim A_HG(n, p, k)$, which now implies that it also does so for $\wG_H \sim G_H(n, p, k)$. But as argued above, finding an independent set of size $k$ in $\wG_H \sim G_H(n, p, k)$ implies that
ALGRAND finds an independent set of size $k'$ in $H \in \cal{H}$, thus solving an NP-hard problem. 
Hence the assumption that there is a polynomial time algorithm ALG that can find independent sets of size $k$ in $G \sim A\oG(n, p, k)$ implies that NP has randomized polynomial time algorithms.

\section{Additional results}
\label{sec:discussion}

In the main part of the paper we only described what we view as our main results. The appendix contains all missing proofs, and some additional results and extensions, not described above. For example, one may ask for which value of $p \le \frac{1}{2}$ the transition occurs from being able to find the maximum independent set in $G \sim A\bar{G}(n, p, k)$ in polynomial time, to the problem becoming NP hard. Our results show a gradual transition. For constant $p$ the problem remains polynomial time solvable, and then, as $p$ continues to decrease, the running time of our algorithms becomes super polynomial, and grows gradually towards exponential complexity. Establishing this type of behavior does not require new proof ideas, but rather only the substitution of different parameters in the existing proofs. Consequently, some theorems that were stated here only in special cases (e.g., \cref{thm:ourtheta}  that was stated only for $p = \frac{1}{2}$) are restated in the appendix in a more general way (e.g., replacing $\frac{1}{2}$ by $p$), and a more general proof is provided.

Though this is not shown in the appendix, our hardness results (for finding adversarially planted independent sets) also imply a gradual transition, 
providing NP-hardness results when $p = n^{\delta-1}$, and as $p$ grows (e.g., into the range $p = \frac{1}{(\log n)^c}$) the NP-hardness results are replaced by hardness results under stronger assumptions, such as  (a randomized version of) the exponential time hypothesis. This is because for $p = \frac{1}{(\log n)^c}$ we need to limit the size of the graphs $H \in {\cal{H}}$ to be only polylogarithmic in $n$, as for larger sizes the proofs in \cref{sec:hardness} fail. 

An interesting range of parameters that remains open is that of $p = \frac{d}{n}$ for some large constant $d$. The case of a random planted independent set of size $\sqrt{\frac{c}{d}} n$ (for some sufficiently large constant $c > 0$ independent of $d$) was addressed in~\cite{FO08}. In such sparse graphs, the planted independent set is unlikely to be the maximum independent set.
The main result in~\cite{FO08} is a polynomial time algorithm that with high probability finds the maximum independent set in that range of parameters. It would be interesting to see whether the positive results extend to the case of adversarial planted independent set. We remark that neither \cref{thres1} nor \cref{thres3} apply in this range of parameters.

%% file: bounding.tex
\section{Bounding the theta function}\label{sec:bounding}

In this section we will prove \cref{thm:ourtheta}.
\begin{theorem}[\cref{thm:ourtheta} restated]
	Let $G \sim AG(n, p, k)$ with $p = 1/2$ and $k = k(n) = 10\sqrt{n}$.
	Then $k \leq \vartheta(\oG) \leq k + 96\log n$ with probability at least $1 - \exp(-2k\log n)$, for every possible choice of $k$ vertices by the adversary.
\end{theorem}
Instead of proving exactly this theorem, we will prove a generalization to other values of $p$.
Let $c \in (0, 1)$ be an arbitrary constant. 
\begin{theorem}\label{Gsemitheta}
    Consider an arbitrary function $w(n)$, such that $n^{2/3} \ll w(n) \leq cn$.
    Let $G \sim AG(n, p, k)$, where $p = w(n)/n$ and $k = Cw(n)^{1/2}$ for constant $C > 0$ large enough ($C = \frac{5}{1 - p}$ suffices).
    Let $a(n, p) := \frac{48}{(1 - p)^2}p\log n$.
    Then $k \leq \vartheta(\oG) \leq k + a(n, p)$, for every possible choice of $k$ vertices by the adversary, with probability at least $1 - \exp(-2k\log n)$.
\end{theorem}
This theorem has a very important corollary, which follows from the Lipschitz property of Lovasz theta function \cite{Lovasz}.
\begin{corollary}\label{Gaddtheta}
    Let $p$ and $k$ be as in \cref{Gsemitheta}, and let $K \subset V$ be the vertices belonging to the planted clique of $G\sim AG(n, p, k)$. 
    Then, with probability at least $1 - \exp(-2k\log n)$,
    \begin{enumerate}[label=(\roman*)]
        \item for every subset $T\subset K$, 
        \[k - |T| \leq \vartheta(\overline{G\setminus T}) \leq k  - |T| + a(n, p),\]
        where $G\setminus T$ denotes the graph $G$ with vertices from $T$ deleted;
        \item for every subset $S\subset V\setminus K$, if we ``add'' $S$ to the planted clique by drawing all edges between $S$ and $S\cup K$, for the resulting graph $G_S$
        \[k + |S| \leq \vartheta(\overline{G_S}) \leq k  + |S| + a(n, p).\]
    \end{enumerate}
\end{corollary}

We now prove \cref{Gsemitheta}.
For $G \sim AG(n, p, k)$, its complement graph $\oG$ contains a planted independent set of size $k$, so $\vartheta(\oG) \geq \alpha(\oG) \geq k$.
It remains to prove the upper bound.
We will use the formulation of the theta function as an eigenvalue minimization problem: \begin{equation}
       \vartheta(G) = \min_M[\lambda_1(M)]
\end{equation}
Here $M$ ranges over all $n$ by $n$ symmetric matrices in which $M_{ij} = 1$ whenever $(i, j) \not\in E$, and $\lambda_1(M)$ denotes the largest eigenvalue of $M$.

The following proposition will be used in the proof of \cref{Gsemitheta}.

\begin{proposition}\label{Gneigh}
    Let $k$ and $p$ be as in \cref{Gsemitheta}.
    Let $G ' \sim G(n, p)$, $G' = (V, E)$.
    Let $\mu p < \nu \leq \mu \in (0, 1]$ be arbitrary constants.
    For any $t \geq 0$, for every set $Q\subset V$ of size $(\mu + o(1))k$, there are at most $g(n, p, \mu, \nu, t) := \frac{6\mu }{(\nu - p\mu)^2}p\log n + t$ vertices from $V\setminus Q$ that have at least $(\nu - o(1))k$ neighbors in $Q$, with probability at least $1 - \exp\left(-\frac{(\nu - p\mu)^2tk}{3\mu p}\right)$.
    Here $o(1)$ is any function of $n$ tending to $0$.
\end{proposition}

\begin{proof}
    For convenience, we will consider the size of $Q$ to be exactly $\mu k$, and consider the set of vertices that have at least $\nu k$ neighbors in $Q$, as addition of $o(1)$-function does not affect anything in the proof. We shall also use $g(n, p, t)$ as shorthand notation for $g(n, p, \mu, \nu, t)$.
    
    Fix some set $Q\subset V$ of size $\mu k$, and a set $I \subset V\setminus Q$ of size $m$. Let $T(I, Q)$ denote the event that every vertex in $I$ has at least $\nu k$ neighbors in $Q$.
	Consider a random bipartite graph with parts $I$ and $Q$ and edge probability $p$, and let $e(I, Q)$ be the number of edges between $I$ and $Q$.
	It is clear that $\E{e(I, Q)} = m\mu kp$, and the event $T(I, Q)$ implies the event $\{e(I, Q) \geq m\nu k\}$.
	Hence
	\begin{multline*}
	\P{T(I, Q)} \leq \P{e(I, Q) \geq m\nu k } = \\
	=\P{e(I, Q) \geq \E{e(I, Q)} + m\cdot (\nu - p\mu)k} \leq \\ \leq
	2\exp\left(-\frac{\left(\nu - p\mu\right)^2mk}{2\mu p}\right).
	\end{multline*}
	There are $\binom{n}{m}\leq \left(\frac{ne}{m}\right)^m\leq \exp(2m\log n)$ possible vertex sets $I$, and $\binom{n}{k} \leq \exp(2k\log n)$ possible subsets $Q$. Let $T_m$ be the event that for at least one such choice of $I$ and $Q$ the event $T(I, Q)$ holds.
	By union bound,
	\[   \P{T_m} \leq 2\exp\left(2k\log n + m \cdot \left(2\log n -\frac{(\nu - p\mu)^2k}{2 \mu p}\right)\right).\]
	Since $k = Cw(n)^{1/2}$, $p = w(n)/n$ and $w(n) = O(n)$, $\frac{k}{p} = \Omega(\sqrt{n})$, so $2\log n -\frac{(\nu - p\mu)^2k}{2 \mu p} \leq -\frac{(\nu - p\mu)^2k}{3 \mu p}$,
	and $\P{T_m} \leq \exp\left(2k\log n - m\cdot\frac{(\nu - p\mu)^2k}{3\mu p}\right)$.
	It is clear that $2k\log n - m\cdot\frac{(\nu - p\mu)^2k}{3\mu p} < 0$ if and only if $m > \frac{6\mu }{(\nu - p\mu)^2}p\log n$,
	so if $m > g(n, p, t) := \frac{6 \mu}{(\nu - p\mu)^2}p\log n + t$, then $\P{T_m} \leq \exp\left(-\frac{(\nu - p\mu)^2tk}{3\mu p}\right)\xrightarrow{n\to\infty}0$.
	Therefore, 
	with probability at least $1 - \exp\left(-\frac{(\nu - p\mu)^2tk}{3\mu p}\right)$ for every set $Q$ of size $\mu k$, there are at most $g(n, p, t) = g(n, p, \mu, \nu, t)$ vertices from $V\setminus Q$ that have at least $\nu k $ neighbors in $Q$.
\end{proof}

By setting $\mu = 1$ and $\nu = \frac{1 + p}{2}$ and choosing $t = \frac{24p}{(1-p)^2}\log n$ we get an immediate corollary.

\begin{corollary}\label{Gcloseneigh}
    With probability at least $1 - \exp(-2k\log n)$ over choice of $G'$, for every $S\subset V$, $|S| = k$, there are at most $a(n, p) := \frac{48}{(1 - p)^2}p\log n$ vertices from $V\setminus S$ with at least $\frac{1 + p}{2}k$ neighbors in $Q$.
\end{corollary}

Let $G = (V, E)$ where $V = [n]$.
Let $Q$ be the set of all vertices with at least $\frac{1 + p}{2}k$ neighbors in the planted clique $K$.
By \cref{Gcloseneigh} $k' := |Q| \leq k + a(n, p)$.
We number the vertices of $V$ in such a way that $Q = [k']$, and the planted $k$-clique of $G$ is $[k]$.

We derive an upper bound on $\vartheta(\oG)$ by presenting a particular matrix $M$, for which $\vartheta(\oG) \leq \lambda_1(M) \leq k' \leq k + a(n, p)$. We use $d(i, Q)$ to denote the number of edges between the vertex $i \in [n]\setminus[k']$ and the set $Q = [k']$.
The symmetric matrix $M$ we choose is as follows. 

\begin{itemize}
    \item The upper left $k' \times k'$ block is all-ones matrix of order $k'$.
    \item The lower right block of size $(n - k') \times (n - k')$, denoted by $C$, is defined as $c_{ij} = 1$ if $(i, j) \in E$ and $c_{ij} = -\frac{p}{1-p}$ if $(i,j) \not\in E$.
    \item The lower left block is an $(n - k') \times k'$ matrix $B$. For this matrix, $b_{ij} = 1$ if $(i, j) \in E$, and $b_{ij} = -d(i, Q)/(k' - d(i, Q))$ if $(i,j) \not\in E$. Observe that that every row of $B$ sums up to zero. 
    \item The upper right block is the transpose of the lower right block $B$.
\end{itemize}
We rewrite $b_{ij}$ for $(i, j) \notin E$ in the following way:
\[    b_{ij} = x_i - \frac{p}{1 - p}, \qquad \qquad  x_i = \frac{k'p - d(i, Q)}{(1 - p)(k' - d(i, Q))}.\]

The vector with $1$ in its first $k'$ entries and $0$ in other $n - k'$ coordinates is an eigenvector of $M$ with eigenvalue $k'$.
To show that $k'$ is the largest eigenvalue, it suffices to prove that $\lambda_2(M) < k'$.	
We represent $M$ as a sum of three symmetric matrices $M = U + V + W$, and apply Weyl theorem \cite{HJ12}:

$$\lambda_2(M) \leq \lambda_1(U) + \lambda_2(V + W) \leq \lambda_1(U) + \lambda_2(V) + \lambda_1(W).$$

Matrices $U$, $V$ and $W$ are as follows.
\begin{itemize}
    \item The matrix $U$ is derived from the adjacency matrix of the original graph $G' \sim G(n,p)$. 
    $U_{ii} = 0$ for all $i$, $U_{ij} = 1$ if $(i, j)\in E$ (in $G'$), and $U_{ij} = -p/(1 - p)$ for all other $i \not =j$.
    \item Matrix $V$ describes the modification that $G'$ undergoes by planting the clique $K$ and extending it to $Q$. For $i, j \leq k'$ we have $V_{ij} = 1 - U_{ij}$, which is $1/(1 - p)$ if $(i, j)$ was not an edge of $G'$. All other entries are~0.
    \item The matrix $W$ is the correction matrix for having the row sums of $B$ equal to~0. In its lower left block ($i > k'$ and $j \le k'$), $W_{ij} = 0$ if $b_{ij} = 1$ and $W_{ij} = x_i$ if $b_{ij} = x_i - p/(1 - p)$. Its upper right block is the transpose of the lower left block. All other entries are~0. 
\end{itemize}

\begin{claim}\label{eigenclaim}
    With probability at least $1 - \exp(-2k\log n)$, for every possible choice of $k$ vertices by adversary, we have
    \[\lambda_1(U) \leq \frac{2 + o(1)}{\sqrt{1 - p}}w(n)^{1/2}, \qquad \lambda_2(V) = o(k'), \qquad \lambda_1(W) \leq \frac{2 + o(1)}{1 - p}w(n)^{1/2}.\]
\end{claim}

To bound the eigenvalues of $U$, $V$ and $W$, we shall use upper bounds on the eigenvalues of random matrices, as appear in~\cite{V07}.
\begin{theorem}\label{eigbound}
	There are constants $C'$ and $C''$ such that the following holds.
	Let $a_{ij}$, $i, j \in [n]$ be independent random variables, each of which has mean $0$ and variance at most $\sigma^2$ and is bounded in absolute value by $L$, where $\sigma \geq C''L\frac{\log^2n}{\sqrt{n}}$.
	Let $A$ be the corresponding $n\times n$ matrix.
	Then with probability at least $1 - O(1/n^3)$,
	\[     \lambda_1(A) \leq 2\sigma\sqrt{n} + C'(L\sigma)^{1/2}n^{1/4}\log n.\]
	The bound holds regardless of what the diagonal elements of $A$ are, since by subtracting the diagonal we may decrease the eigenvalues at most by $L$.
\end{theorem}
The matrix $U$ is a random matrix, as it is generated from the graph $G' \sim G(n, p)$.
The entries of matrix $U$ have mean zero, $|U_{ij}| = O(1)$ since $p$ is bounded by constant $c < 1$, and the variance is $\sigma^2(U_{ij}) = \E{U_{ij}^2} = p/(1 - p) \gg n^{-1/3}$, 
so by \cref{eigbound} we have 
\[\lambda_1(U) \leq 2\sqrt{\frac{np}{1 - p}} + O\left(\left(\frac{np}{1 - p}\right)^{1/4}\log n\right) \leq \frac{2}{\sqrt{1 - p}}w(n)^{1/2} + O\left(w(n)^{1/4}\log n\right) =: \Lambda_U\]
with probability at least $1 - O(1/n^3)$.
Since $|U_{ij}| = O(1)$ for all $i, j \in [n]$, $\lambda_1(U)$ is at most $O(n^2)$.
Then, the expected value of $\lambda_1(U)$ is at most $\E{\lambda_1(U)} \leq \Lambda_U + O(1/n)$.
It follows that $\P{\lambda_1(U) \geq \Lambda_U + t} \leq \P{\lambda_1(U) \geq \E{\lambda_1(U)} + t}$ for all non-negative $t$.
Hence, to show that $\lambda_1(U)$ does not exceed $\lambda_U$ by too much with extremely high probability, it suffices to show that the probability of $\lambda_1(U)$ to deviate from its mean is exponentially small in $k \log n \simeq w(n)^{1/2}\log n$.
The result by Alon, Krivilevich and Vu \cite{AKV02} ensures that eigenvalues of $U$ are well-concentrated around their means.
\begin{theorem}[Concentration of eigenvalues]\label{eigcon}
	For $1 \leq i \leq j \leq n$, let $a_{ij}$ be independent, real random variables with absolute value at most 1.
	Define $a_{ji} = a_{ij}$ for all $i, j$, and let $A$ be the $n \times n$ matrix with $A(i, j) = a_{ij}$, $i, j \in [n]$.
	Let $\lambda_1(A) \geq \lambda_2(A) \geq \ldots \geq \lambda_n(A)$ be the eigenvalues of $A$.
	For all $s \in [n]$ and for all $t = \omega(\sqrt{s})$:
	\[    \P{\left|\lambda_s(A) - \E{\lambda_s(A)}\right| \geq t} \leq \exp\left(-\frac{(1 - o(1))t^2}{32s^2}\right).\]
	The same estimate holds for $\lambda_{n - s + 1}(A)$.
\end{theorem}
Taking $t = \Theta(w(n)^{1/4}\log n)$, from \cref{eigcon} we get
\begin{multline*}
    \P{\lambda_1(U) \geq \frac{2}{\sqrt{1 - c}}w(n)^{1/2} + \Theta\left(w(n)^{1/4}\log n\right)} \leq\\
    \leq \P{\lambda_1(U) \geq \E{\lambda_1(U)} + \Theta\left(w(n)^{1/4}\log n\right)} \leq \exp\big(-\Omega\big(w(n)^{1/2}\log^2n\big)\big),
\end{multline*}
so $\lambda_1(U) \leq \frac{2}{\sqrt{1 - c}}w(n)^{1/2} + O\left(w(n)^{1/4}\log n\right)$
with probability at least $1 - \exp\big(-\Omega\big(k\log^2n\big)\big)$.
Note that the bound holds for any choice of the adversary, as matrix $U$ does not depend on the vertices of the planted clique and is determined by initial graph $G(n, p)$ only.


As for the matrix $V$, we shift it so that all its entries have mean 0.
Precisely, we consider matrix $V'$ such that for all $i, j > k'$ we have $V'_{ij} = V_{ij} = 0$, $V'_{ii} = 0$ for $i \in [k']$, and for $i < j \leq k'$ we have $V'_{ij} = V_{ij} - 1$, which is either $-1$ with probability $p$ and $p/(1- p)$ with probability $(1 - p)$.
Basically, $V'$ is a copy of matrix $U$ of order $k'$, so from \cref{eigbound} we can obtain the bounds for $\lambda_2(V)$, which is $\lambda_1(V')$.
The variance is $\sigma^2(V'_{ij}) = p/(1 - p) \gg \frac{\log^4n}{n}$, so with probability at least $1 - O(1/k'^3)$
\[\lambda_1(V') \leq 2\sqrt{\frac{k'p}{1 - p}} + O\left(\left(\frac{k'p}{1 - p}\right)^{1/4}\log k'\right) \leq  \frac{C'w(n)^{3/4}}{\sqrt{n}} + O\left(\frac{w(n)^{3/8}}{n^{1/4}}\log n\right)\]
for some constant $C' > 0$, we will denote this bound by $\Lambda_{V'}$.
Similarly to $\lambda_1(U)$, we have $\lambda_1(V') \leq O(k'^2)$ and $\E{\lambda_1(V')} \leq (1 + o(1))\Lambda_{V'}$. 
Applying \cref{eigbound} to $\lambda_1(V')$ with $t \gg  \frac{w(n)^{3/8}}{n^{1/4}}\log n$, we get $\P{\lambda_2(V) > \frac{C'w(n)^{3/4}}{\sqrt{n}} + t} \leq \exp\left(-\Omega\left(t^2\right)\right)$.
We would like these bounds hold for any choice of the adversarial $k$-clique.
There are $\binom{n}{k} \leq \left(\frac{ne}{k}\right)^k \leq \exp(2k\log n) \leq \exp(O(w(n)^{1/2}\log n))$ possible choices, so by setting $t =\Theta(w(n)^{1/4}\log n)$ in the bound above and applying union bound over all possible choices of $k$-clique, we prove
\[  \lambda_2(V) \leq \frac{C'w(n)^{3/4}}{\sqrt{n}} + O\left(w(n)^{1/4}\log n\right) = o(k')\]
for any choice of the adversary with probability at least $1 - \exp\big(-\Omega\big(k\log^2n\big)\big)$.

It remains to bound $\lambda_1(W)$.
We will use the trace of $W^2$.
\[\lambda_1(W)^2 \leq \tr(W^2) = 2\sum_{i< j}W_{ij}^2 = 2\sum_{i = k' + 1}^n(k' - d(i, Q))x_i^2 = 2\sum_{i = k' + 1}^n\frac{\left(d(i, Q) - k'p\right)^2}{(1 - p)^2(k' - d(i, Q))}.\]
By definition of set $Q$, for every $k' + 1 \leq i \leq n$ we have $d(i, Q) \leq \frac{1 + p}{2}k$, so $k' - d(i, Q) \geq \frac{1 - p}{2}k$ and
\[   2\sum_{i = k' + 1}^n\frac{\left(d(i, Q) - k'p\right)^2}{(1 - p)^2(k' - d(i, Q))} \leq \frac{4}{(1 - p)^3k}\sum_{i = k' + 1}^n\left(d(i, Q) - k'p\right)^2.\]
It turns out that we can always bound the sum above.
\begin{theorem}\label{proba}
    With probability at least $1 - \exp(-2k\log n)$,
	\[   \sum_{i = k' + 1}^n\left(d(i, Q) - k'p\right)^2 \leq (n - k')k'p(1 - p) + o(nk'p(1 - p))\]
	for every possible choice of $k$ vertices by the adversary,
\end{theorem}
The proof is rather technical and is presented in \cref{sec:probbound}.
From \cref{proba} we get
\[   \lambda_1(W)^2 \leq  \frac{4 + o(1)}{(1 - p)^3k}(n - k')k'p(1 - p) \leq \frac{4 + o(1)}{(1 - p)^2}(n - k')p,\]
so 
\[    \lambda_1(W) \leq \frac{2 + o(1)}{1 - p}\sqrt{(n - k')p} \leq \frac{2 + o(1)}{1 - p}w(n)^{1/2}.\]
Combining the bounds for $\lambda_1(U)$, $\lambda_2(V)$ and $\lambda_1(W)$, we get
\begin{multline*}
	\lambda_2(M) \leq \lambda_1(U)  + \lambda_2(V) + \lambda_1(W) \leq \\
	\leq \frac{2}{\sqrt{1 - p}}w(n)^{1/2} + \frac{2}{1 - p}w(n)^{1/2} + o(k') \leq \frac{4}{1 - p}w(n)^{1/2} + o(k').
\end{multline*}
By choosing $C \geq \frac{5}{1 - p}$ in $k = Cw(n)^{1/2}$, we guarantee that the expression above is less than $k'$.
Therefore, $k'$ is indeed the largest eigenvalue of matrix $M$, and $\vartheta(\oG) \leq k' \leq k + a(n, p)$ for every choice of adversarial $k$-clique with extremely high probability.
This finishes the proof of \cref{Gsemitheta}.

%% file: mainalgo.tex
\section{Main algorithm}\label{sec:mainalgo}

In this section we prove \cref{thres1}.

\begin{theorem}[\cref{thres1} restated]
	For every fixed $\eps > 0$ and for every $k \geq \eps\sqrt{n}$, there is an (explicitly described) algorithm running in time $n^{O(\log(\frac{1}{\eps}))}$ which almost surely finds the maximum clique in a graph $G \sim AG(n, \frac{1}{2}, k)$. The statement holds for every adversarial planting strategy (choice of $k$ vertices as a function of $G' \sim G(n, \frac{1}{2})$), and the probability of success is taken over the choice of  $G' \sim G(n, \frac{1}{2})$.
\end{theorem}
As with \cref{thm:ourtheta} and \cref{Gsemitheta}, we will prove a more general version of the theorem, considering $G \sim AG(n, p, k)$ for a wide range of values of $p$, and not just $p= \frac{1}{2}$. We first prove such a theorem when $k \geq C\sqrt{np}$ for a sufficiently large constant $C$. Afterwards, we shall extend the proof to the case that $C$ can be an arbitrarily small constant.

\begin{theorem}\label{Gsemialgo}
    Let $c \in (0, 1)$ be an arbitrary constant. 
    Consider an arbitrary function $w(n)$, such that $n^{2/3} \ll w(n) \leq cn$.
    Let $G \sim AG(n, p, k)$, where $p = w(n)/n$ and $k \geq \frac{5}{1 - p}w(n)^{1/2}$.
    There is an (explicitly described) algorithm running in time $n^{O(1)}$ which almost surely finds the maximum clique in $G$, for every adversarial planting strategy.
\end{theorem}

\begin{proof}
As described in \cref{sec:algorithm}, we solve the optimization problem
\begin{equation}\label{thetaP}
       \vartheta(G) = \max_{h,\{s_i\}} \sum_{i\in V} (h \cdot s_i)^2,
\end{equation} 
finding the optimal orthonormal representation $\{s_i\}$ and handle $h$, using the SDP formulation.

Suppose that we solved $\vartheta(\oG)$ in (\ref{thetaP}) for $G\sim AG(n, p, k)$ (with $p$ and $k$ as in \cref{Gsemialgo}).
By \cref{Gsemitheta}, 
$k \leq \vartheta(\oG)\leq k + a(n, p)$.
Let $G = (V, E)$, let $K$ denote the set of vertices chosen by the adversary.

As $h$ and $s_i$ are unit vectors, we have that for all $i \in V$, $(h \cdot s_i)^2 \leq 1$.
Let $T$ be the set of vertices $i \in K$ with $(h \cdot s_i)^2 < 3/4$.
We claim that $|T| \leq 4a(n, p)$.
Suppose the contrary, so $|T| > 4a(n, p)$. 
Delete $|T|$ from the graph $G$ and consider $\vartheta(\overline{G\setminus |T|})$.
We get
\[\vartheta(\overline{G\setminus T}) \geq \sum_{i \in V\setminus T}(h\cdot s_i)^2 = \vartheta(\oG) - \sum_{j \in T}(h\cdot s_j)^2,\]
hence by applying \cref{Gaddtheta} to $G\setminus T$ we get
\begin{multline*}
\vartheta(\oG) \leq \vartheta(\overline{G\setminus T}) + \sum_{j \in T}(h\cdot s_j)^2 \leq k - |T| + a(n, p)  + (3/4)|T| =\\
= k + a(n, p) - (1/4)|T|< k + a(n, p) -  a(n, p) = k,
\end{multline*}
a contradiction.
So, there are at most $4a(n, p)$ vertices $j \in K$ with $(h \cdot s_j)^2 < 3/4$, implying that there are at least $k - 4a(n, p)$ vertices in $K$ with $(h \cdot s_j)^2 \geq 3/4$.
Denote this set by $K_{3/4}$.

Observe that if $i \in V\setminus K$ is not connected to some $j \in K_{3/4}$, then $(h\cdot s_i)^2 \leq 1/4$.
Indeed, $(i, j)\notin E$ implies $s_i \cdot s_j = 0$, so $(h\cdot s_i)^2 + (h\cdot s_j)^2 \leq 1$ and therefore $(h\cdot s_i)^2 \leq 1/4$.
Hence, if $i \in V\setminus K$ has $(h\cdot s_i)^2\geq 3/4$, it must be connected to the whole set $K_{3/4}$.
The set $K_{3/4}$ has size at least $k - 4a(n, p)$, so by 
\cref{Gcloseneigh} there are less than $a(n, p)$ vertices $i \in V\setminus K$ with $(h\cdot s_i)^2 \geq 3/4$.
As a result, for the set $H$ of vertices $i \in V$ with $(h\cdot s_i)^2 \geq 3/4$, we have $k - 4a(n, p) \leq |H| \leq k + a(n, p)$.

Let $F \subset V$ be the set of all vertices that have at least $3k/4$ neighbors in $H$.
Similarly to \cref{lem:setF}, with extremely high probability $F$ contains the maximum clique in $G$. Moreover, by \cref{Gneigh} there are at most $O(a(n, p))$ vertices from $V\setminus H$ that have at least $3/4k$ neighbors in $H$, implying $|F| \leq k + O(a(n, p))$.

If follows that the maximum clique of $G[F]$, the subgraph of $G$ induced on $F$, is the maximum clique of $G$.
Moreover, $K\subseteq F$, so $F$ contains a clique of size at least $k$, and $|F| \leq k + O(a(n, p))$.
The maximum clique in $G[F]$ can be found in polynomial time by a standard algorithm (used for example to show that vertex cover is fixed parameter tractable). For every non-edge in the subgraph induced on $F$, at least one of its end-vertices needs to be removed, so we try both possibilities in parallel, and recurse on each subgraph that remains. Each branch of the recursion is terminated either when the graph is a clique, or when $k$ vertices remain (whichever happens first). At least one of the branches of the recursion finds the maximum clique.
The depth of the recursion is at most $O(a(n, p)) = O\left(\frac{p}{(1 - p)^2}\log n\right)$. Consequently the running time (which is exponential in the depth) is in the order of $n^{O(1)} 2^{O(a(n,p))} = n^{O(1/(1-p)^2)}$. This running time is polynomial if $p$ is upper bounded by a constant smaller than~1. 
This finishes the description of the algorithm, proving \cref{Gsemialgo}.
\end{proof}




We now return to \cref{thres1}, which considers $G \sim AG(n, \frac{1}{2}, k)$ and $k \geq \eps\sqrt{n}$. By plugging in $p = \frac{1}{2}$ in \cref{Gsemialgo}, we prove \cref{thres1} when $\epsilon \ge \frac{10}{\sqrt{2}}$, as \cref{Gsemialgo} assumes the condition $k \geq \frac{5}{1 - p}w(n)^{1/2}$ (where $w(n) = np$). To prove \cref{thres1} we need to handle arbitrarily small constant $\epsilon > 0$. For this, we extend the proof of \cref{Gsemialgo} to handle the case that $k \geq \eps w(n)^{1/2}$ for arbitrarily constant $\epsilon > 0$. 

Suppose that $k = \eps\sqrt{np}$ for $0 < \eps < \frac{5}{1 - p}$.
Similar to the approach of \cite{AKS98}, we can use the algorithm that works for the case $k \geq \frac{5}{1 - p}\sqrt{np}$ in order to obtain the algorithm for $k = \eps\sqrt{np}$.

Let $s$ be the smallest integer satisfying $k = \eps\sqrt{np} > 2\cdot \frac{5}{1-p}\sqrt{np}\cdot p^{s/2}$. This gives $s > \frac{2\log \frac{10}{(1-p)\eps}}{\log (1/p)}$, which is a constant for constant $\eps > 0$ and $p$ bounded away from~1. Observe that if $p < \frac{\eps^2}{100}$ then $s = 1$. Given graph $G \sim AG(n, p, k)$, we try all $n\choose s$ possible choices for sets $S\subset V$ of size $s$. For each such choice, if $S$ is a clique in $G$, then we apply the algorithm of \cref{Gsemialgo} on $G[N(S)]$, the subgraph induced on the common neighborhood of $S$ (not including $S$ itself). The size of this subgraph is at most roughly $p^s n + k$, and we chose the value of $s$ so that $k - s \ge \frac{5}{1-p}\sqrt{|N(S) p|}$. As we show below, for all $k \choose s$ choices in which $S\subset K$, the algorithm will return the largest clique in $G[N(S)]$. As the largest clique $K^*$ of $G$ contains at least $s$ vertices of $K$, for at least one choice of $S \subset K$ we also have $S \subset K^*$. For this case, the union of $S$ and the largest clique in $G[N(S)]$ is the largest clique of $G$, as desired.

It remains to show that if $S \subset K$, then the algorithm of \cref{Gsemialgo} finds the maximum clique in $G[N(S)]$. In more details, what we need to show is that with high probability over the choice of $G' \sim G(n, p)$, for every choice of $k$ vertices as the adversarial planted clique $K$ (giving the graph $G$), and for every choice of $S \subset K$ of size $s$, the algorithm succeeds on $G[N(S)]$.

We shall employ a union bound over all possible choices of $S$ and $K$. Given that we consider all possible $K$ (and not just the one selected by the adversary), we may describe the generation of $G[N(S)]$ in the following way. 

\begin{enumerate}
    \item Start with the empty graph on a set $V$ of $n$ vertices.
    \item Pick a set $K \subset V$ of $k$ vertices, and a set $S \subset K$ of $s$ vertices.
    \item Generate $G' \sim G(n, p)$ in a need to know basis. First reveal only those edges between $S$ and $V \setminus S$. Let $V'$ denote the set of vertices that each has all of $S$ as its neighbors. Observe that the expected size of $V'$ is exactly $\E{|V'|} = (n-k)p^s$. 
    \item Form the set $N(S) = V' \cup (K \setminus S)$. We now reveal the edges of $G'$ inside the set $N(S)$, giving a graph that we call $G'_{S,K}$. Crucially, this graph is distributed exactly like $G(|N(S)|, p)$.
    \item Turn $K \setminus S$ into a clique in $G'_{S, K}$, effectively planting a clique of size $k - s = k - O(1)$.
\end{enumerate}

\begin{claim}
    With probability at least $1 - \exp(-2k\log n)$ over the choice of $G' \sim G(n, p)$, for all possible choices of $K \subset V$ and $S \subset K$ it holds simultaneously that $|N(S)| = (1 + o(1))np^s$. 
\end{claim}
\begin{proof}
    Fix some particular choices of $K$ and $S$.
    By construction, set $N(S)$ is a union of $V'$ and vertices from $K\setminus S$.
    We are going to show that
    \begin{itemize}
        \item for $k = \eps\sqrt{np}$, it holds $k\log n = o\left(np^s\right)$;
        \item with probability at least $1 - \exp(-2k\log n)$, $|N(S)| \leq np^s + k + 3\sqrt{np^sk\log n}$.
    \end{itemize}
    Given these two statements, the claim follows directly.
    
    First consider the case $p < \eps^2/100$, so $s = 1$.
    We can assume that $p \gg n^{-1/3}$, as when $p = O(n^{-1/3})$ we can find the maximum clique using the enumeration algorithm from \cref{sec:enumeration}.
    Then $k = O(\sqrt{np}) = O(n^{1/2})$, while $np^s = np = \Omega(n^{2/3})$, therefore $k\log n = o(np^s)$ holds.
    Since $\E{|V'|} = (n - k)p^s$, by Chernoff bound the probability of $|V'| > (n - k)p^s + 3\sqrt{(n - k)p^sk\log n}$ is at most $\exp(-2k\log n)$, and we get the desired.
    
    Now consider the case $\eps^2/100 \leq p \leq c$, so $p$ is a constant.
    But then $k\log n = O(\sqrt{n}\log n) = o(n) = o(np^s)$ since $s$ is a constant, and Chernoff bound again gives us $|V'| \leq np^s + 3\sqrt{np^sk\log n}$ with probability at least $1 - \exp(-2k\log n)$.
    
    
    Uniting $V'$ with $K\setminus S$ cannot add more than $|K| = k$ vertices, therefore $|N(S)| \leq |V'| + |K| \leq np^s + k + 3\sqrt{np^sk\log n}$, with probability at least $1 - \exp(-2k\log n)$.
\end{proof}

By the above claim and our choice of $s$  we now have that $k -s > \frac{5}{1-p}\sqrt{|N(S)|p}$, where $k-s$ is the size of the clique planted in $G'_{S, K}$. Consequently, we are in a position to apply 
\cref{Gsemialgo} on $G[N(S)]$, and conclude that the algorithm given in the proof of the theorem finds the maximum clique in $G[N(S)]$. This indeed holds almost surely for every particular choice of $K\subset V$ and $S\subset K$, but we are not done yet, as we want this to hold for all choices of $K$ and $S$ in $G' 
\sim G(n, p)$. To reach such a conclusion we need to analyse the failure probability of \cref{Gsemialgo} more closely, so as to be able to take a union bound over all choices of $K$ and $S$. This union bound involves ${n \choose k} \cdot {k \choose s} \simeq \exp(k\log n)$ events (the term  ${k \choose s}$ is negligible compared to ${n \choose k}$, because $s$ is a constant).

Indeed the failure probability for \cref{Gsemialgo} can withstand such a union bound. This is because the proof of \cref{Gsemialgo} is based on earlier claims whose failure probability is at most $\exp(-2k\log n)$. This upper bound on the failure probability is stated explicitly in  \cref{Gsemitheta} and \cref{Gcloseneigh}, and can be shown to also hold in claims that do not state it explicitly (such as \cref{pro:maxIS}, \cref{lem:setH} and \cref{lem:setF}, and versions of them  generalized to arbitrary $p$), using analysis similar to that of the proof of \cref{Gneigh}.

%% file: hardness_new.tex
\section{Hardness}\label{sec:newhardnessproof}

\subsection{Maximum Independent Set in balanced graphs}

\begin{definition}
	Given a graph $H$, denote its average degree by $\alpha$.
	A graph $H$ is balanced if every induced subgraph of $H$ has average degree at most $\alpha$.
\end{definition}

\begin{theorem}\label{nph}
	For any $0 < \eta \leq 1$, determining the size of the maximum independent set in a balanced graph with average degree $2 < \alpha < 2 + \eta$ is NP-hard.
\end{theorem}
\begin{proof}
	It is well known that given a parameter $k$ and a 3-regular graph $H$, determining whether $H$ has an independent set of size $k$ is NP-hard. For simplicity of upcoming notation, let $2n$ denote the number of vertices in $H$. Given a positive integer parameter $t$, we describe a polynomial time reduction $\Rr$ such that given a 3-regular graph $H$ it holds that: 
	\begin{itemize}
		\item $\Rr(H)$ is a balanced graph with average degree $2 + \frac{1}{3t + 1}$.
		\item $\Rr(H)$ has an independent set of size $k + 3nt$ if and only if $H$ has an independent set of size $t$.
	\end{itemize}
	
	By choosing $t > \frac{1}{3\eta - 6}$, the theorem is proved.
	
	Let $H$ be a 3-regular graph on $2n$ vertices. The graph $\Rr(H)$ is obtained from $H$ 
	by replacing every edge $(u,v)$ of $H$ by a path with $2t$ intermediate vertices that connects between $u$ and $v$.
	There are $3n$ edges in $H$, so by doing so we add $2t \cdot 3n$ vertices of degree $2$.
	The average degree of the resulting graph $\Rr(H)$ is
	\[    \alpha = \frac{2t \cdot 3n \cdot 2 + 2n \cdot 3}{2t \cdot 3n + 2n} = \frac{6t + 3}{3t + 1} = 2 + \frac{1}{3t + 1} \]
	as desired.

	To see that the graph $\Rr(H)$ is balanced, consider a subset of vertices $S^*\subseteq \Rr(H)$, and let $\alpha^* > 2$ denote the average degree of the induced subgraph $\Rr(H)[S^*]$.
	W.l.o.g., we can assume that $\Rr(H)[S^*]$ has minimum degree at least 2 (because if $\Rr(H)[S^*]$ has a vertex of degree at most~1, removing it would result in a subgraph of higher average degree). Let $V_3$ be the set of vertices of degree~3 in $\Rr(H)[S^*]$. All remaining vertices of $\Rr(H)[S^*]$ have degree~2. As no two degree~3 vertices in $\Rr(H)$ are neighbors, $\Rr(H)[S^*]$ is composed of degree~3 vertices, and non-empty disjoint paths connecting between them. As no path connecting two degree~3 vertices in $\Rr(H)$ has fewer than $2t$ vertices (it may have more than $2t$ vertices, if it goes through original vertices of $H$), the number of degree~2 vertices in $\Rr(H)[S^*]$ is at least $\frac{3|V_3|}{2}\cdot 2t$. Hence $\alpha^* \le 2 + \frac{1}{3t + 1}$, as desired.
	
	Every independent set $I$ of size $k$ in $H$ gives rise to an independent set of size $k + 3nt$ in $\Rr(H)$, because in $\Rr(H)$ we can take the vertices of $I$ and $t$ vertices from each of the $3n$ length $t$ paths (at least one of the two end vertices of each path is not adjacent to a vertex in $I$). Likewise, every independent set of size $k + 3nt$ in $\Rr(H)$ gives rise to an independent set of size $k$ in $\Rr(H)$. Note that $I$ contains at most $t$ vertices from any single path of $\Rr(H)$, and moreover, can be assumed to contain exactly $t$ vertices from any single path of $\Rr(H)$ (if $I$ contains fewer than $t$ vertices from the path connecting $u$ and $v$, then by taking all even vertices of the path one gains a vertex, and this compensates for the at most one vertex that is lost from $I$ due to the possible need to remove $v$ from $I$). As $I$ contains $3nt$ path vertices, its remaining $k$ vertices are from $H$. Moreover, they form an independent set in $H$ (no two vertices $u$ and $v$ adjacent in $H$ can be in this set, because then the path connecting them in $\Rr(H)$ cannot contribute $t$ vertices to $I$).
\end{proof}

\subsection{Notation to be used in the proof of \cref{thm:closeStatistics}}
\label{sec:notation}

In the coming sections we prove \cref{thm:closeStatistics}. For simplicity of the presentation (and without affecting the implications towards the proof of \cref{thres3}), we describe the distributions $G_H(n, p)$, $G_H(n, p, k)$ and $A_HG(n, p, k)$ in a way that differs from their description in \cref{sec:hardness}. Based on these descriptions, we will present $X_H(G)$, a key random variable associated with these distributions. This random variable is easier to work with than the random variable referred to in \cref{lem:manyH}, and hence we shall later slightly change the formulation of \cref{lem:manyH} (without affecting the correctness of  \cref{thm:closeStatistics}). 

It will be convenient for us to think of $G$ as an $n$ vertex graph with vertices numbered from~1 to $n$, and of $H$ as an $m$ vertex graph with vertices numbered from~1 to $m$. For simplicity, we assume that $m$ divides $n$ (this assumption can easily be removed with only negligible effect on the results).  Given an $n$-vertex graph $G$, we partition the vertex set of $G$ into $m$ disjoint subsets of vertices, each of size $\frac{n}{m}$. Part $i$ for $1 \le i \le m$ contains the vertices $[(i-1)\frac{n}{m} + 1, i\frac{n}{m}]$. A vertex set $S$ of size $m$ that contains one vertex in each part is said to \textit{obey the partition}. 

\begin{definition}
\label{def:G_H}
    Let $H$ be an arbitrary $m$-vertex graph, and let $n$ be such that $m$ divides $n$, let $k' \le m$ be a parameter (specifying the conjectured size of the maximum independent set in $H$), and let $k$ satisfy $k' \le k \le n - m$. We say $G_H$ is distributed by $G_H(n, p)$ (for $p \in (0, 1)$) and that $\wG_H$ is distributed by $G_H(n, p, k)$ if they are created by the following random process. 
    \begin{enumerate}
    \item Generate a random graph $G' \sim G(n, p)$, with a partition of its vertex set into $m$ parts.    
        \item Choose a random subset $M$ of $m$ vertices from $G'$ that obeys the partition.
        \item For every $1 \le i \le m$, associate vertex $i$ of $H$ with the vertex of $M$ in the $i$th part, and replace the induced subgraph of $G'$ on $M$ by the graph $H$. This gives $G_H \sim G_H(n,  p)$.
        \item Within the non-neighbors of $M$, plant at random an independent set $I'$ of size $k-k'$, giving the graph $\wG_H \sim G_H(n, p, k)$. (If $M$ has fewer than $k - k'$ non-neighbors in $G_H$, an event that will happen with negligible probability for our choice of parameters, then we say that this step {\em fails}, and instead we plant a random independent set of size $k$ in $G_H$.)
    \end{enumerate}
\end{definition}

Though the description is different, it is not difficult to show that the distributions $G_H(n, p)$ and $G_H(n, p, k)$ are identical to the corresponding distributions described in \cref{sec:hardness}.

We also change the description of distribution $A_HG(n, p, k)$ from \cref{sec:hardness} in a way analogous to the above, by fixing a partition of the vertices of $G' \sim G(n,  p)$ and requiring the adversary to choose in $G'$ an induced copy of $H$ that obeys the partition (vertex $i$ of $H$ must be in part $i$ of the partition, for every $1 \le i \le m$). As in \cref{sec:hardness}, the adversary also plants a random independent set of size $k - k'$ among the non-neighbors of $H$. If either $G'$ does not have an induced copy of $H$ that obeys the partition, of there are too few non-neighbors of $H$, we say that the adversary {\em fails}, and we revert to the default procedure of planting a random independent set of size $k$ in $G'$.

We note that there is a (negligible) difference in the probability of failure in the above description of $A_HG(n,  p, k)$ compared to that of \cref{sec:hardness}, because it might be that $G'$ has an induced copy of $H$, but no induced copy of $H$ that obeys the partition.

For a graph $G$ and a given partition, $X_H(G)$ denotes the number of sets $S$ of size $m$ obeying the partition, such that the subgraph of $G$ induced on $S$ is $H$ (with vertex $i$ of $H$ in part $i$ of the partition, for every $1 \le i \le m$).
For a graph $G$ chosen at random from some distribution, $X_H(G)$ is a random variable.

\subsection{Proof of \cref{lem:manyH}}
\label{sec:proofLem2.3}

As noted in \cref{sec:notation}, we slightly change \cref{lem:manyH}. Instead of referring to all induced copies of $H$, we refer only to induced copies of $H$ that obey the partition. The random variable $X_H(G)$ denotes their number. The main technical content of this modified \cref{lem:manyH} is handled by the following lemma.

\begin{lemma}\label{copy}
    Let $0 < \eps < 1/7$ be a constant, and let $G \sim G(n, p)$ be a random graph with $p \in (0, 1)$.
    Let $H$ be a balanced graph on $m$ vertices with average degree $2 < \alpha < 3$.
	If $m \le \min[\sqrt{\frac{\eps}{p}}, 2^{-1/4}p^{\alpha/4}\sqrt{\eps n}]$ (or equivalently, $\eps \geq m^2p$ and $\eps^2 \geq 2\frac{m^4}{n^2p^\alpha}$), then for every $\beta \in [0, 1)$
	\[     \P{X_H(G) \leq \beta\E{X_H(G)}} \leq \frac{4\eps}{(1 - \beta)^2}.\]
\end{lemma}
\begin{proof}	
    Let $w(n) := np$, so $p = w(n)/n$.
	Let $Y_H(G)$ be a random variable counting the number of sets $S$ obeying the partition that have $H$ as an edge induced subgraph of $G$, but may have additional internal edges.
	By definition, $X_H(G) \leq Y_H(G)$ and
	\[   \E{Y_H(G)} = \left(\frac{n}{m}\right)^mp^{\frac{\alpha m}{2}} = \left(\frac{n}{m}\right)^m\cdot \left(\frac{w(n)}{n}\right)^{\frac{\alpha m}{2}} = \left(\frac{w(n)^{\alpha}}{m^2n^{\alpha - 2}}\right)^{\frac{m}{2}}.\]
	A set $S$ in $Y_H(G)$ contributes to $X_H(G)$ if it has no internal edges beyond those of $H$.
	This happens with probability 
	\[   (1 - p)^{\binom{m}{2} - \frac{\alpha m}{2}} \geq 1 - \binom{m}{2}\frac{w(n)}{n} \geq 1 - \frac{m^2w(n)}{n} \geq 1 - \eps,\]
	so
	\[   \E{X_H(G)} \geq (1 - p)^{\binom{m}{2} - \frac{\alpha m}{2}}\E{Y_H(G)} \geq (1 - \eps)\E{Y_H(G)}.\]
	
	Note that if $\E{Y_H(G)^2} \leq (1 + \eps)\E{Y_H(G)}^2$, the inequality above gives us $\E{X_H(G)^2} \leq \frac{1 + \eps}{(1 - \eps)^2}\E{X_H(G)}^2$.
	We will now compute $\E{Y_H(G)^2}$.
	Given the occurrence of $H$, consider another potential occurrence $H'$ that differs from it by $t$ vertices.
	Since $H$ is balanced graph,
	\[|E(G[H'])| - |E(G[H'\cap H])| \geq \frac{\alpha |V(G[H'])|}{2} - \frac{\alpha|V(G[H\cap H'])|}{2} \geq \frac{\alpha t}{2}.\]
	Hence, the probability that $H'$ realized conditioned on $H$ being realized is at most $p^{\frac{\alpha t}{2}}$.
	The number of ways to choose $t$ other vertices is $\binom{m}{t}\left(\frac{n}{m}\right)^t$ (first choose $t$ groups out of $m$ in the partition, then choose one vertex in each group).
	Hence, the expected number of such occurrences is 
	\[    \mu_t \leq \binom{m}{t}\left(\frac{n}{m}\right)^tp^{\frac{\alpha t}{2}} = \binom{m}{t}\left(\frac{w(n)^{\alpha}}{m^2n^{\alpha - 2}}\right)^{\frac{t}{2}}. \]
	It follows that $\mu_m \leq \E{Y_H(G)}$.
	Moreover,
	\[ \sum_{t = 1}^{m/2}\frac{\mu_t}{\E{Y_H(G)}} \leq \sum_{t = 1}^{m/2}\binom{m}{t}\left(\frac{w(n)^{\alpha}}{m^2n^{\alpha - 2}}\right)^{\frac{t - m}{2}} \leq 2^m\left(\frac{w(n)^{\alpha}}{m^2n^{\alpha - 2}}\right)^{-\frac{m}{4}} = \left(\frac{16m^2n^{\alpha - 2}}{w(n)^{\alpha}}\right)^{\frac{m}{4}}.\]
	Furthermore,
	\[  \sum_{t  = m/2}^{m - 1}\frac{\mu_t}{\E{Y_H(G)}} \leq \sum_{t = m/2}^{m - 1}m^{m - t}\left(\frac{w(n)^{\alpha}}{m^2n^{\alpha - 2}}\right)^{\frac{t - m}{2}} = \sum_{t  = m/2}^{m - 1}\left(\frac{w(n)^{\alpha}}{m^4n^{\alpha - 2}}\right)^{\frac{t - m}{2}}.\]
	When $w(n)^\alpha \geq 2m^4n^{\alpha - 2}$ the term $t = m - 1$ dominates, and hence the sum is at most roughly $\sqrt{\frac{m^4n^{\alpha - 2}}{w(n)^{\alpha}}}$.
	Since $\eps^2 \geq 2\frac{m^4n^{\alpha  - 2}}{w(n)^\alpha}$ we have 
    \[
	    \sum_{t = 1}^{m  - 1}\frac{\mu_t}{\E{Y_H(G)}} \leq \left(\frac{16m^2n^{\alpha - 2}}{w(n)^{\alpha}}\right)^{\frac{m}{4}} + (1 + o(1))\sqrt{\frac{m^4n^{\alpha - 2}}{w(n)^{\alpha}}}
	    \leq (1 + o(1))\sqrt{\frac{m^4n^{\alpha - 2}}{w(n)^{\alpha}}} \leq \eps\]
	and $\sum_{t = 1}^m\mu_t \leq (1 + \eps)\E{Y_H(G)}$.
	Hence can bound $\E{Y_H(G)^2} \leq \E{Y_H(G)}\sum_{t = 1}^m\mu_t \leq (1 + \eps)\E{Y_H(G)}^2$, and
	\[\E{X_H(G)^2} \leq \E{Y_H(G)^2} \leq (1 + \eps)\E{Y_H(G)}^2 \leq \frac{1 + \eps}{(1 - \eps)^2}\E{X_H(G)}^2 \leq (1 + 4\eps)\E{X_H(G)}^2.\]
	The last inequality holds since $\eps < 1/7$.
	We get that $\V{X_H(G)} = \E{X_H(G)^2} - \E{X_H(G)}^2 \leq 4\eps\E{X_H(G)}^2$.
	By Chebyshev's inequality we conclude that
	\begin{multline*}
	    \P{X_H(G) \leq \beta\E{X_H(G)}} = \P{X_H(G) \leq \E{X_H(G)} - (1 -\beta)\E{X_H(G)}} \leq\\\leq  \frac{\V{X_H(G)}}{(1 - \beta)^2\E{X_H(G)}^2}\leq \frac{4\eps}{(1 -\beta)^2},
	\end{multline*}
	as desired.
\end{proof}
\begin{corollary}[\cref{lem:manyH} restated]\label{subcopy}
	For every $0 < \delta < 1$, $0 < \eps < 1/7$, $2 < \alpha < \min(\frac{2}{1 - \delta}, 3)$ and $0 < \rho < \min[\frac{1 - \delta}{2}, \frac{2 - \alpha(1 - \delta)}{4}]$ the following holds for large enough $n$.
	Let $G \sim G(n, p)$ be a random graph with $p = n^{\delta - 1}$, and let $H$ be a balanced graph on $m = n^{\rho}$ vertices and with average degree $\alpha$. Then $\E{X_H(G)}\xrightarrow{n\to\infty}+\infty$, and for every $\beta \in [0, 1)$
	\[     \P{X_H(G) \leq \beta\E{X_H(G)}} \leq \frac{4\eps}{(1 - \beta)^2}.\]
\end{corollary}
\begin{proof}
We first note that $\alpha < \frac{2}{1 - \delta}$ implies that $2 - \alpha(1 - \delta) > 0$, and hence we can take $\rho > 0$ in the above Corollary.
	The inequality $\rho < (1 - \delta)/2$ implies (for large enough $n$) that 
	\[ m = n^{\rho} < \sqrt{\frac{\eps}{n^{\delta - 1}}} = \sqrt{\frac{\eps}{p}} \]
	Likewise, $\rho < (2 - \alpha(1 - \delta))/4$ implies (for large enough $n$) that
	\[ 2m^4 \le \eps^2 n^{2 - \alpha(1 - \delta)} = \eps^2 n^2 p^{\alpha} \]
	The above bounds on $m$ satisfy the requirements of \cref{copy}, and hence  
		\[     \P{X_H(G) \leq \beta\E{X_H(G)}} \leq \frac{4\eps}{(1 - \beta)^2}.\]
		
	To show that $\E{X_H(G)}\xrightarrow{n\to\infty}+\infty$, recall the notation $w(n) = np$ and the following bound from the proof of \cref{copy} 
	\[   \E{X_H(G)}\geq (1 - \eps) \left(\frac{w(n)^{\alpha}}{m^2n^{\alpha - 2}}\right)^{\frac{m}{2}}.\]
	As $w(n) = n^\delta$ and $m = n^\rho$,
	\[  \frac{w(n)^{\alpha}}{m^2n^{\alpha - 2}} = n^{-2\rho - \alpha(1 - \delta) + 2}.\]
	Since $\rho < (2 - \alpha(1 - \delta))/4$, $2 - \alpha(1 - \delta) - 2\rho > 2\rho$, so $\E{X_H(G)} \geq (1 - \eps)n^{m\rho} \xrightarrow{n\to\infty}+\infty$.
\end{proof}

\subsection{Proofs of \cref{lem:closeStatistics} and \cref{thm:closeStatistics}}
\label{sec:proofThm2.3}

\begin{lemma}[\cref{lem:closeStatistics} restated]\label{closedist}
Let $p(G)$ denote the probability to output $G$ according to $G(n, p)$, and let $p_H(G)$ denote the probability to output $G$ according to $G_H(n, p)$.
	For every constant $\beta \in [0, 1)$, with probability at least $1 - \frac{4\eps}{(1 - \beta)^2}$ over the choice of graph $G \sim G(n, p)$, it holds that $p_H(G) \geq \beta p(G)$.
\end{lemma}
\begin{proof}
	Let $e$ be the number of edges in $G$ and consider $p_H(G)$.
	Out of the $\left(\frac{n}{m}\right)^m$ options to choose a subset $M$ in $G_H(n, p)$, only $X_H(G)$ options are such that the subgraph induced on $M$ is $H$, so that the resulting graph could be $G$.
	Since $H$ has average degree $\alpha$, it has exactly $\alpha m / 2$ edges.
	Note that 
	\[     \E{X_H(G)} = \left(\frac{n}{m}\right)^mp^{\frac{\alpha m}{2}}(1 - p)^{\binom{m}{2} - \frac{\alpha m}{2}}.\]
	Given that we chose a suitable $M$, the rest of the edges $(e - \frac{\alpha m}{2})$ of $G_H(n, p)$ should agree with $G$.
	It follows that 
	\begin{multline*}
	p_H(G) = \frac{X_H(G)}{\left(\frac{n}{m}\right)^m}p^{e - \frac{\alpha m}{2}}(1 - p)^{\binom{n}{2}- \left(\binom{m}{2} + e - \frac{\alpha m}{2}\right)} = \\
	= \frac{X_H(G)}{\left(\frac{n}{m}\right)^mp^{\frac{\alpha m}{2}}(1 - p)^{\binom{m}{2} - \frac{\alpha m}{2}}}p^e(1 - p)^{\binom{n}{2} - e} = \\=\frac{X_H(G)}{\E{X_H(G)}}p^e(1 - p)^{\binom{n}{2} - e} = \frac{X_H(G)}{\E{X_H(G)}}p(G).
	\end{multline*}
	By \cref{subcopy}, for every $\beta \in [0, 1)$,
    \[ \P[G\sim G(n, p)]{X_H(G) \leq \beta\E{X_H(G)}} \leq \frac{4\eps}{(1 - \beta)^2}.\]
    It follows that for every $G$ with $X_H(G) \geq \beta\E{X_H(G)}$ we have 
	\[  p_H(G) = \frac{X_H(G)}{\E{X_H(G)}}p(G) \geq \frac{\beta\E{X_H(G)}}{\E{X_H(G)}}p(G) = \beta p(G).\]
	Therefore, by \cref{subcopy}, for every $\beta \in [0, 1)$, for at least $1 - \frac{4\eps}{(1 - \beta)^2}$ fraction of all graphs $G \sim G(n, p)$ we will have $p_H(G) \geq \beta p(G)$.
\end{proof}

We now restate and prove \cref{thm:closeStatistics}. Recall that now $G_H(n, p, k)$ and $AG_H(n, p, k)$ refer to the distributions as defined in \cref{sec:notation}, rather that those defined in \cref{sec:hardness}.  

\begin{theorem}[\cref{thm:closeStatistics} restated]\label{thm:closedist}
Let $f$ be an arbitrary function that gets as input an $n$ vertex graph and outputs either $0$ or $1$. Let $p_A$ denote the probability that $f(G) = 1$ when $G \sim A_HG(n, p, k)$, and let $p_H$ denote the probability that $f(G) = 1$ when $G \sim G_H(n, p, k)$.
	For every constant $\beta \in [0, 1)$, it holds that $p_H \geq \beta (p_A - \frac{4\eps}{(1 - \beta)^2})$.
\end{theorem}

\begin{proof}
    For clarity of the analysis, let break into small steps the computation of $f(G)$ when $G \sim A_HG(n, p, k)$. 
    
    \begin{enumerate}
        \item Generate a graph $G' \sim G(n, p)$.
        \item Choose in $G'$ a random induced copy of $H$ that obeys the partition. If there is no such induced copy this step is said to {\em fail}, and one invokes the {\em default} (explained in item~4).
        \item Plant at random an independent set of size $k - k'$ among the non-neighbors of the chosen induced copy of $H$, giving the graph $G$. If this induced copy has fewer than $k - k'$ non-neighbors, this step is said to {\em fail}, and one invokes the {\em default} (explained in item~4).
        \item If the default is invoked, plant at random an independent set of size $k$ in $G'$, giving the graph $G$.
        \item Compute $f(G)$.
    \end{enumerate}
    
    The computation of $f(G)$ when $G \sim G_H(n, p, k)$ is {\em identical} to the above, except that step~2 is replaced by the following:

    \begin{itemize}
        \item Plant in $G'$ at random an induced copy of $H$ that obeys the partition, giving the graph $G_H \sim G_H(n,p)$. 
    \end{itemize}

    Call a graph $G$ {\em typical} if the probability of generating $G'$ under $G_H(n, p)$ is at least $\beta$ times the probability of generating $G$  under $G(n, p)$. Let $T$ denote the event the a graph $G' \sim G(n, p)$ is typical. By \cref{closedist}, $\P{T} \ge 1 - \frac{4\eps}{(1 - \beta)^2}$. Let $F_A$ denote the event that $f(G) = 1$ for a graph $G \sim A_HG(n, p, k)$, and recall that $\P{F_A} = p_A$. Let $T \cap F_A$ denote the coupled event that both $T$ and $F_A$ happen, when the respective $G'$ is the outcome of step~1 in the generation of the respective $G \in A_HG(n, p, k)$. If follows that $\P{T \cap F_A} \ge p_A - \frac{4\eps}{(1 - \beta)^2}$. 
    
    Observe that given that a graph $G'$ is typical, then step~2 of the process of generating  $G \sim A_HG(n, p, k)$ does not fail. Moreover, if the same graph $G'$ is obtain as $G_H$ in step~2 of the generation of $G \sim G_H(n, p, k)$, then afterwards the process of generating $G \sim G_H(n, p, k)$ is {\em identical} to that of generating $G \sim G_H(n, p, k)$. (This uses the fact that for any such $G'$, each of the induced copies of $H$ that obey the partition has exactly the same probability of being the planted one under $G_H(n, p)$.) Hence the event of generating from this $G_H$ a graph $G \in G_H(n, p, k)$ for which $f(G) = 1$ is {\em exactly} the same event as that of generating from the respective $G'$ a graph $G \in G_H(n, p, k)$ for which $f(G) = 1$. As for every {\em typical} graph the probability of generating it under $G_H(n, p)$ is at least $\beta$ times the probability of generating it under $G(n, p)$, we conclude that $p_H \ge \beta \P{T \cap F_A} \ge \beta(p_A - \frac{4\eps}{(1 - \beta)^2})$.
\end{proof}

\subsection{Proof of \cref{thres3}}

To prove \cref{thres3} we shall use \cref{thm:closedist} together with a few relatively simple lemmas. 
\cref{nonneigh} implies that the probability that $G_H(n, p, k)$ fails to produce an output graph is negligible. (For $A_HG(n, p, k)$, the same is implied by the combination of \cref{nonneigh} and  \cref{subcopy}.) 

\begin{lemma}\label{nonneigh}
    Let $\delta \in (0, 1)$ and $0 < \rho < (1-\delta)/2$.
    Let $G \sim G(n, p)$ be a random graph with $p = n^{\delta - 1}$.
    For every set $S$ of $n^{\rho}$ vertices of $G$ the size of the common non-neighborhood of $S$ is at least $n - 2n^{\delta + \rho}$ with probability at least $1 - \exp(-\frac{1}{4}n^{\delta})$.
\end{lemma}
\begin{proof}
    We clearly have $\rho + \delta < 1$.
    By Chernoff bound, the maximum degree of $G$ is at most $2n^{\delta}$ with probability at least $1 - \exp(-\frac{1}{4}n^{\delta})$.
    Hence, any set $S$ of $n^{\rho}$ vertices has at most $2n^{\delta + \rho}$ neighbors.
    Then, for any set $S$ of this size the common non-neighborhood of $S$ has size at least $n - 2n^{\delta + \rho}$ with probability at least $1 - \exp(-\frac{1}{4}n^{\delta})$.
\end{proof}

The following lemmas establish that with high probability the graph $G \sim G_H(n, p, k)$  has no independent set that has more than $k - k'$ vertices outside the induced copy of $H$. The notation used in these lemmas is as in \cref{def:G_H}.

\begin{lemma}
\label{lem:noOtherIS}
Let $p = n^{\delta - 1}$ and $k - k' \geq 4n^{1 - \delta}\log n$.
	With probability at least $1 - \exp(-\frac{1}{2}n^{1 - \delta}\log^2 n)$, there is no independent set of size $\frac{k - k'}{2}$ in $G_H[V\setminus M]$.
\end{lemma}
\begin{proof}
	By first moment method the probability that there exists an independent set of size $t$ is at most
	\begin{multline*}
	    \binom{n}{t}(1 - p)^{\frac{t(t - 1)}{2}} \leq \exp\left(t\log n + t - t\log t -\frac{t(t - 1)}{2}n^{\delta - 1}\right) =\\= \exp\left(\log n + 1 - \log t - \frac{t - 1}{2}n^{\delta - 1}\right)^t,
	\end{multline*}
	which for $t \geq \frac{k - k'}{2} \geq 2n^{1 - \delta}\log n$ is at most $\exp(-\frac{1}{2}n^{1 - \delta}\log^2 n)\xrightarrow{n\to\infty}0$.
\end{proof}

\begin{lemma}
\label{lem:noIntersectingIS}
Let $p = n^{\delta - 1}$ and $k - k' \geq 6n^{1 - \delta}\log n$.
	For every integer $t$ satisfying  $1 \le t \le \frac{k-k'}{2}$, with probability at least $1 - 2/n$ every subset $Q \subset V\setminus (M\cup I')$ of vertices of graph $\wG_H[V\setminus M]$, $|Q| = t < k - k'$, has at least $t + 1$ neighbors in $I'$. 
\end{lemma}
\begin{proof}
    To prove this, view the process of generating $\wG_H$ in a following way.
	Initially, we have the graph $H$ and $n - m$ isolated vertices.		Then, for every pair of vertices $u, v$ where $u \in M$ and $v \in V\setminus M$, draw an edge $(u, v)$ with probability $p$.
	By doing so, we determine the set $W\subseteq V\setminus M$ of vertices that have no neighbors in $H$.
	Select a random subset $I' \subset W$ of size $k - k'$. For every pair of vertices from $V\setminus M$, if at least one of them does not belong to $I'$, draw an edge with probability $p$.
	
	There are at most $\binom{n}{t} \leq n^t$ possible choices for the set $Q$. There are at most $\binom{k-k'}{t} \leq k^t \le n^t$ possible choices for the set $Y$ of at most $t$ neighbors of $Q$ within $I'$. The probability that $Q$ has no neighbors in $I' \setminus Y$ is $(1-p)^{t(k-k'-t)} \le (1-p)^{t(k-k')/2} \le n^{-3t}$. By a union bound the probability that some subset $Q\subset V\setminus (M\cup I')$ of size $t$ has at most $t$ neighbors in $I'$ is at most $n^{-t}$. 	The probability of this happening for some value $t \le \frac{k - k'}{2}$ is at most $\sum_{t = 1}^{(k - k')/2}n^{-t} \leq \frac{2}{n}$, as desired.
\end{proof}

Combining the above lemmas we have the following Corollary. 

\begin{corollary} \label{cor:feasibleStrategy}
Let $p = n^{\delta - 1}$ and $6n^{1 - \delta}\log n \le k - k' \le \frac{2n}{3}$. Then with probability at least $1 - 4/n$ over the choice of graph $G \sim G_H(n, p, k)$, every independent set of size $k$ in $G$ contains at least $k'$ vertices in the planted copy of $H$.
\end{corollary}

\begin{proof}
There are three events that might cause the Corollary to fail.

\begin{itemize}
    \item $G_H(n, p, k)$ fails to produce an output. By \cref{nonneigh} and the upper bound on $k$, the probability of this event is smaller than $\frac{1}{n}$.
    \item Even before planting $I'$, there is an independent set larger than $\frac{k - k'}{2}$ in $G_H[V\setminus M]$. By \cref{lem:noOtherIS} the probability of this event is smaller than $\frac{1}{n}$.
    \item After planting $I'$, one can obtain an independent set larger than $I'$ in $G_H[V\setminus M]$ by combining an independent set $Q \subset V\setminus (M\cup I')$ with some of the vertices of $I'$. As we already assume that \cref{lem:noOtherIS} holds, $Q$ can be of size at most $\frac{k-k'}{2}$. \cref{lem:noIntersectingIS} then implies that the probability of this event is at most $\frac{2}{n}$.
\end{itemize}

The sum of the above three failure probabilities is at most $\frac{4}{n}$. 
\end{proof}

Now we restate and prove \cref{thres3}.

\begin{theorem}
For $p = n^{\delta-1}$ with $0 < \delta < 1$, $0 < \gamma < 1$, and $6n^{1 - \delta}\log n \le k \le \frac{2}{3}n$ the following holds.
There is no polynomial time algorithm that has probability at least $\gamma$ of finding an independent set of size $k$ in $G \sim A\bar{G}(n, p, k)$, unless NP has randomized polynomial time algorithms (NP=RP). 
\end{theorem}


\begin{proof}
Suppose for the sake of contradiction that algorithm ALG has probability at least $\gamma$ of finding an independent set of size $k$ in the setting of the Theorem. 

Choose $2 < \alpha < \min[\frac{2}{1 - \delta} , 3]$ and $0 < \rho < \min[\frac{1 - \delta}{2}, \frac{2 - \alpha(1 - \delta)}{4}]$. Let $\cal{H}$ be the class of balanced graphs of average degree $\alpha$ on $m = n^{\rho}$ vertices. By \cref{nph}, given a graph $H \in \cal{H}$ and a parameter $k'$, it is NP-hard to determine whether $H$ has an independent set of size $k'$. We now show how ALG can be leveraged to design a randomized polynomial time algorithm that solves this NP-hard problem with high probability.

Repeat the following procedure $10\frac{\log n}{\gamma}$ times.

\begin{itemize}
    \item Sample a graph $G \sim G_H(n, p, k)$ (as in \cref{def:G_H}).
    \item Run ALG on $G$. If ALG returns an independent set of size $k$ that has at least $k'$ vertices in the planted copy of $H$, then answer {\em yes} ($H$ has an independent set of size  $k'$) and terminate.
\end{itemize}

If $10\frac{\log n}{\gamma}$ iterations are completed without answering {\em yes}, then answer {\em no} ($H$ probably does not have an independent set of size  $k'$).

Clearly, the above algorithm runs in random polynomial time. Moreover, if it answers {\em yes} then its answer is correct, because it actually finds an independent set of size $k'$ in $H$. It remains to show that if $H$ has an independent set of size $k'$, the probability of failing to give a {\em yes} answer is small.

We now lower bound the probability that a single run of ALG on $G \sim G_H(n, p, k)$ fails to output {\em yes}. Recall that ALG succeeds (finds an independent set of size $k$) with probability at least $\gamma$ over graphs with adversarially planted independent sets, and in particular, over the distribution $A_HG(n, p, k)$. 

In \cref{subcopy}, choose $\eps = \frac{\gamma}{25}$ and $\beta = \frac{1}{5}$. Our choice of $m = n^{\rho}$ satisfies the conditions of \cref{copy}, and hence we can apply \cref{thm:closedist}.
In \cref{thm:closedist} use the function $f$ that has value~1 if ALG succeeds on $G$. It follows from \cref{thm:closedist} that ALG succeeds with probability at least $\beta(\gamma - \frac{4\eps}{(1 - \beta)^2})=\frac{3\gamma}{20}$ over graphs $G \sim G_H(n, p, k)$. \cref{cor:feasibleStrategy} implies that there is probability at most $\frac{4}{n}$ that there is an independent set of size $k$ in $G$ that does not contain $k'$ vertices in the induced copy of $H$. Hence  a single iteration returns {\em yes} with probability at least $\frac{3\gamma}{20} - \frac{4}{n} \ge \frac{\gamma}{10}$ (for sufficiently large $n$).


Finally, as we have $10\frac{\log n}{\gamma}$ iterations, the probability that none of the iterations finds an independent set of size $k$ is at most $(1 - \frac{\gamma}{10})^{10\frac{\log n}{\gamma}} \simeq \frac{1}{n}$.
\end{proof}

%% file: probbound.tex
\section{Probabilistic bound} \label{sec:probbound}

In this section we prove \cref{proba}.

Let $c \in (0, 1)$ and $C > 0$ be arbitrary constants.
Let $G \sim G(n, p)$, $G = (V, E)$, where $p = w(n)/n$ for $\log^{4} n \ll w(n) < cn$, and let $k = Cw(n)^{1/2}$.
Let $K \subset V$ be arbitrary, $|K| = k$.
We number the vertices of $G$ so that $V = [n]$, $K = [k]$ and $V\setminus K = [n]\setminus [k]$.
For $k + 1 \leq i \leq n$ let $X_i$ be a random variable equal to the number of edges from $i$ to vertices in $K$.
It is clear that $X_i \sim \Bin(k, p)$, so $\E{X_i} = kp$ and the variance $\V{X_i} = \E{(X_i - kp)^2} = kp(1 - p)$.
Since for $i \neq j$, $X_i$ and $X_j$ are independent, 
\[  \V{\sum_{i = k + 1}^n X_i} = \sum_{i = k + 1}^n\V{X_i} = (n - k)kp(1 - p) = \E{\sum_{i = k + 1}^n\left(X_i - kp\right)^2}.\]
Our goal is to show that the sum $\sum_{i = k + 1}^n\left(X_i - kp\right)^2$ does not exceed its mean too much. 
\cref{varbound} is a restatement of \cref{proba}, with somewhat different notation.
\begin{theorem}\label{varbound}
    With probability at least $1 - \exp\left(-2k\log n \right)$,
	\[    \sum_{i = k + 1}^n\left(X_i - kp\right)^2 \leq (n - k)kp(1 - p) + o(nkp(1 - p))\]
	for every possible choice of the set $K\subset V$.
\end{theorem}

We prove \cref{varbound} in several steps.
Let $U_i = \left(X_i - kp\right)^2$, $\E{U_i} = kp(1 - p)$.
We need to prove that the value $\sum_{i = k + 1}^nU_i$ doesn't deviate from its mean, $(n - k)kp(1 - p)$, too much.
However, the maximum possible value of $U_i$ is $k^2(1 - p)^2$, which can be close to $k^2$.

Partition all vertices $k + 1 \leq i \leq n$ into $R$ groups, defined by the following rules.
For $r \leq R - 1$ the vertex $i$ belongs to the group $M_r$, if $\frac{k^2}{2^r} \leq U_i \leq \frac{k^2}{2^{r - 1}}$.
If $U_i$ is at least $k^2 \cdot 2^{-r}$, then $X_i$ differs from $kp$ by at least $k2^{-r/2}$, and if $U_i$ is at most $k^2\cdot 2^{-r + 1}$, then $X_i$ differs from $kp$ by at most $k2^{-(r - 1)/2}$.
This means that if $i \in M_r$ for $r \leq R - 1$, either
\[     k\left(p + 2^{-r/2}\right) \leq X_i \leq k\left(p + 2^{-(r - 1)/2}\right) \quad \text{or}\quad k\left(p - 2^{-(r - 1)/2}\right) \leq X_i \leq k\left(p - 2^{-r/2}\right)\] must hold.
For $r = R$ the group $M_R$ contains all the remaining vertices, those $i$ for which $U_i \leq \frac{k^2}{2^{R - 1}}$.
The exact value of $R$ will be determined later, and will depend on $w(n) = np$.

We can rewrite the sum above based on the group partitioning:
\[  \sum_{i = k + 1}^n\left(X_i - kp\right)^2 = \sum_{i = k + 1}^nU_i 
= \sum_{r = 1}^{R}\sum_{i \in M_r}U_i
\leq  \sum_{i \in M_{R}}U_i + \sum_{r = 1}^{R - 1}|M_r|\cdot \frac{k^2}{2^{r- 1}}\]
where the last inequality follows from the definition of $M_r$.
We will show that $\sum_{i \in M_{R}}U_i \le (n - k)kp(1 - p) + o(nkp(1 - p))$, and that $\sum_{r = 1}^{R}|M_r|\cdot k^22^{-(r - 1)} \le o(nkp(1 - p))$, with high probability.

We start with the second sum.
Since $k = O(\sqrt{np})$ and $1 - p = O(1)$, it suffices to show that for any choice of $K$, 
$ \sum_{r = 1}^{R - 1}\frac{|M_r|}{2^{r - 1}} = o\left(k\right)$. Note that the failure probability $\exp\big(-\Omega(nw(n)^{-1/4})\big)$ in \cref{levels} is negligible compared to the error probability $\exp\left(-2k\log n \right)$ allowed in \cref{varbound} (for our choice of $w(n)$ and $k$). 

\begin{lemma}\label{levels}
	Denote $\wR := \log Cn- \log(w(n)^{1/2}\log n)$.
	For all $r \leq \wR - 1$,
	\[      |M_r| \leq 2^{r + 2}w(n)^{1/4}\]
	with probability at least $1 - \exp\big(-\Omega(nw(n)^{-1/4})\big)$, for every choice of $K$.
\end{lemma}

\begin{proof}
	Let $M_r = M_r' \sqcup M_r''$, where $i \in M_r'$ if $k\left(p + 2^{-(r - 1)/2}\right) \geq X_i \geq k\left(p + 2^{-r/2}\right)$ and $i \in M_r''$ if $k\left(p - 2^{-(r - 1)/2}\right) \leq X_i \leq k\left(p - 2^{-r/2}\right)$.
	Fixing $|M_r| = m_r$ is equivalent to fixing $|M_r'| = m_r'$ and $|M_r''| = m_r''$ where $m_r' + m_r'' = m_r$.
	For the sake of simplicity, $m_r'' = 0$ and $m_r' = m_r$, so $M_r' = M_r$.
	Let $I_r$ be a fixed set of vertices from $[n]\setminus[k]$, of size $m_r$.
	We are going to bound the probability $\P{M_r = I_r}$.
	Consider a random bipartite subgraph $B(I_r, K, p)$, where one part is $I_r$  and another part is $K$.
	Let $e_r$ be the number of edges in $B(I_r, K, p)$, it is clear that $\E{e_r} = m_rkp$.
	Since $M_r = M_r'$, by definition of $M_r'$, $e_r \geq m_rk\left(p + 2^{-r/2}\right)$.
	So, the event $M_r = I_r$ implies in the event $e_r \geq m_r'k\left(p + 2^{-r/2}\right)$, hence by Chernoff bound
	\[
	\P{M_r = I_r}
	\leq \P{e_r \geq m_rk\left(p + 2^{-r/2}\right)} \leq 2\exp\left(-\frac{\left(m_rk\right)^2}{2^{r + 1}m_rkp}\right) \leq 2\exp\left(-\frac{Cm_rn}{2^{r + 1}w(n)^{1/2}}\right).
	\]
	There are $\binom{n - k}{m_r}$ possible choices of the set $I_r$, so by union bound the probability $\P{|M_r| = m_r}$ is at most
	\begin{multline*}
	\binom{n - k}{m_r}\P{M_r = I_r} \leq\\\leq \exp\left(m_r\log(n - k) + m_r - m_r\log m_r\right)\cdot 2\exp\left(-\frac{Cm_rn}{2^{r + 1}w(n)^{1/2}}\right)\leq \\
	\leq2\exp\left(m_r\cdot \left(\log n + 1 - \log m_r - \frac{Cn}{2^{r + 1}w(n)^{1/2}}\right)\right).
	\end{multline*} 
	Observe that when $r \leq\wR - 1 = \log Cn- \log(w(n)^{1/2}\log n) - 1$ the value under the exponent, $\log n + 1 - \log m_r - \frac{Cn}{2^{r + 1}w(n)^{1/2}}$,
	is at most $1 -\log m_r$, which approaches $-\infty$ as long as $m_r  \to +\infty$.
	
	Let's find the largest possible value of $m_r$ for which the event $|M_r| = m_r$ might happen at least for one choice of $K$, at least for some value of $r \leq\wR - 1$.
	There are exactly $\binom{n}{k}$ possible choices of the set $K$ and the total of $\wR - 1$ groups, so by union bound we need to find the biggest $m_r$ for which $(\wR - 1)\binom{n}{k}\P{|M_r| = m_r}$ does not converge to zero.
	Since $\binom{n}{k} \leq \left(\frac{ne}{k}\right)^k\leq\exp(2k\log n)$ and $\wR - 1\leq \exp(\log\log Cn)$, it is enough to find the smallest $m_r$ for which
	\[  3k\log n \leq O(w(n)^{1/2}\log n)\ll m_r\cdot \left(\frac{n}{2^{r + 1}w(n)^{1/2}} + \log m_r - \log n - 1\right).\]
	
	Suppose that $m_r > 2^{r + 1}w(n)^{1/4}$ for $r \leq\wR - 1$.
	For $r = 1$, $m_r > 4w(n)^{1/4}$, and:
	\begin{multline*}
	m_r\cdot \left(\frac{n}{2^{r + 1}w(n)^{1/2}} + \log m_r - \log n - 1\right) >\\
	>4w(n)^{1/4}\left(\frac{n}{4w(n)^{1/2}} + \frac{1}{4}\log\log w(n) - \log n\right) =\\
	= \frac{n}{w(n)^{1/4}} + 
	w(n)^{1/4}\log\log w(n) - 4w(n)^{1/4}\log n \gg w(n)^{1/2}\log n,
	\end{multline*}
	as $w(n)= O(n)$, so $\frac{n}{w(n)^{1/4}} = \Omega(n^{3/4})$.
	For $r =\wR- 1 = \log Cn- \log(w(n)^{1/2}\log n) - 1$, $m_r > w(n)^{1/4}\cdot \frac{Cn}{w(n)^{1/2}\log n} = \frac{Cn}{w(n)^{1/4}\log n}$ and (by the bound above)
	\begin{multline*}
	m_r\cdot \left(\frac{n}{2^{r + 1}w(n)^{1/2}} + \log m_r - \log n - 1\right) > m_r(\log m_r - 1)>\\
	>  \frac{Cn}{w(n)^{1/4}\log n}\cdot \left(\log Cn - \frac{1}{4}\log w(n) - \log\log n - 1\right) \gg w(n)^{1/2}\log n,
	\end{multline*}
	since $w(n)= O(n)$ and $w(n) < n$, so $\frac{n}{w(n)^{1/4}\log n}\cdot \left(\log n - \frac{1}{4}\log w(n) - \log\log n - 1\right) = \Omega(n^{3/4})$.
	Since $2^{r + 1}w(n)^{1/4}$ is monotone and continuous in $r$, we get that for all $1 \leq r \leq\wR - 1$ if $m_r > 2^{r + 1}w(n)^{1/4}$ then 
	\[k\log n + k - k\log k + \log\log n \leq 3k\log n \ll m_r\cdot \left(\frac{n}{2^{r + 1}w(n)^{1/2}} + \log m_r - \log n - 1\right),\]
	which means that $(\wR - 1)\binom{n}{k}\P{|M_r| = m_r} \leq \exp\big(-\Omega(nw(n)^{-1/4})\big) \xrightarrow{n\to\infty}0$.
	In other words, the probability that there exists such choice of $k$-subset and such $1 \leq r \leq\wR - 1$ that for the corresponding set of vertices $M_r$ we have $|M_r| = |M_r'| > 2^{r + 1}w(n)^{1/4}$ tends to zero.
	
	Earlier we assumed that $M_r = M_r'$, but in general $M_r = M_r' \sqcup M_r''$, and $m_r = m_r' + m_r''$.
	The opposite case is $M_r = M_r''$, and the analysis transfers without any changes, and $|M_r''| \leq 2^{r + 1}w(n)^{1/4}$ with probability at least $1 - \exp\big(-\Omega(nw(n)^{-1/4})\big)$.
	Hence, with probability of at least $1 - \exp\big(-\Omega(nw(n)^{-1/4})\big)$ for every choice of $K$ and every $1 \leq r \leq\wR - 1$ we have $|M_r| = |M_r'| + |M_r''| \leq 2^{r + 2}w(n)^{1/4}$.
\end{proof}

Since $w(n) \gg \log^{4}n$, $\log n \ll w(n)^{1/4}$, and we set the number of groups $R = \wR = \log Cn- \log(w(n)^{1/2}\log n)$.
By \cref{levels}, with probability at least $1 - \exp\big(-\Omega(nw(n)^{-1/4})\big)$,
\[\sum_{r = 1}^{R - 1}\frac{|M_r|}{2^{r - 1}} \leq \sum_{r = 1}^{R - 1}\frac{2^{r + 2}w(n)^{1/4}}{2^{r - 1}} \leq 8R \cdot w(n)^{1/4} 
=O(\log n \cdot w(n)^{1/4}) = o\left(w(n)^{1/2}\right) = o\left(k\right).\]

Now we move to the first sum, for $i \in M_R$ with $R = \wR$ we have $U_i \leq \frac{k^2}{2^{R - 1}} = 2\frac{k^2w(n)^{1/2}\log n}{Cn} = 2\frac{k \cdot Cw(n)\log n}{Cn} = 2kp\log n$.
We need to prove that with extremely high probability for any choice of $k$-subset  $\sum_{i \in M_R}U_i \leq (n - k)kp(1 - p) + o(nkp(1 - p))$.

We will do this by applying the Bernstein inequality \cite{B46}.
\begin{theorem}[Simple form of Bernstein inequality]\label{bernstein}
	Let $Z_1, \ldots, Z_n$ be independent random variables, $\E{Z_i} = 0$ for $1 \leq i \leq n$.
	Suppose that $|Z_i| \leq L$ 
	for all $1 \leq i \leq n$.
	Then, for all $t > 0$,
	\[    \P{\sum_{i = 1}^nZ_i > t} \leq 2\exp\left(-\frac{\frac{1}{2}t^2}{\sum_{i = 1}^n\E{Z_i^2} + \frac{1}{3}Lt}\right).\]
\end{theorem}

By definition of $M_R$, $\sum_{i \in M_R}U_i \leq \sum_{i = k + 1}^n\min\left(U_i, 2kp\log n\right)$. It is clear that for all $i \in M_R$, $\E{\min(U_i, 2kp\log n)} \leq \E{U_i} = kp$.
Also, since $0 \leq \min(U_i, 2kp\log n) \leq U_i$ almost surely,  \[\V{\min(U_i, 2kp\log n)} \leq \E{\min(U_i, 2kp\log n)^2} \leq \E{U_i^2} = \E{(X_i - kp)^4} \leq  k^2p.\]
The last inequality holds because $\E{(X_i - kp)^4}$ is the fourth central moment of a binomial random variable, and as such its value is known to be $kp(1-p)(1 + (3k-6)p(1-p)) \le kp(1-p)(1 + \frac{3k-6}{4})$.

Let $Z_i := \min(U_i, 2kp\log n) - \E{\min(U_i, 2kp\log n)}$ for all $k + 1 \leq i \leq n$, then $\E{Z_i} = 0$ and $\E{Z_i^2} = \V{\min(U_i, 2kp\log n)} \leq k^2p$.
Moreover, $|Z_i| \leq 2kp\log n$.
Recall that $w(n) \gg \log^4 n$, let $\gamma(n) := \sqrt{\frac{w(n)^{1/2}}{13C\log n}}$.
By \cref{bernstein}:
\begin{multline*}
\P{\sum_{i = k + 1}^n\min\left(U_i, 2kp\log n\right) > (n - k)kp(1 - p) + \frac{(n - k)kp(1 - p)}{\gamma(n)}} = \\
= \P{\sum_{i = k + 1}^n\min\left(U_i, 2kp\log n\right) > \sum_{i = k + 1}^n\E{U_i} + \frac{(n - k)kp(1 - p)}{\gamma(n)}} \leq \\
\leq \P{\sum_{i = k + 1}^n\min\left(U_i, 2kp\log n\right) > \sum_{i = k + 1}^n\E{\min(U_i, 2kp\log n)} + \frac{(n - k)kp(1 - p)}{\gamma(n)}} = \\
= \P{\sum_{i = k + 1}^nZ_i >  \frac{(n - k)kp(1 - p)}{\gamma(n)}} \leq\qquad \qquad \qquad \qquad \qquad \qquad 
\end{multline*}
\begin{multline*}
\qquad \qquad \qquad\leq 2\exp\left(-\frac{(n - k)^2k^2p^2(1 - p)^2}{2\gamma(n)^2\left(\sum_{i = k+1}^n\E{Z_i^2} + kp\log n \cdot \frac{(n - k)kp(1 - p)}{3\gamma(n)}\right)}\right)\leq  \\
\qquad 
\leq 2\exp\left(-\frac{(n - k)^2k^2p^2}{2\gamma(n)^2\left((n - k)k^2p + (n - k)k^2p^2\frac{\log n}{3\gamma(n)}\right)}\right) =\\= 2\exp\left(-\frac{(n - k)p}{2\gamma(n)^2\left(1 + \frac{p\log n}{3\gamma(n)}\right)}\right).
\end{multline*}
As $\frac{p\log n}{3\gamma(n)} = O\left(\frac{w(n)^{3/4}\log^{3/2}}{n}\right) = o(1)$, $(n - k)p \simeq w(n)$, $\gamma(n) = \sqrt{\frac{w(n)^{1/2}}{13C\log n}}$, and  $k = Cw(n)^{1/2}$, we have:
\[     \frac{(n -k)p}{2\gamma(n)^2\left(1 + \frac{p\log n}{3\gamma(n)}\right)} \geq \frac{(n - k)p}{4\gamma(n)^2} \simeq \frac{13C\log n \cdot w(n)}{4\sqrt{w(n)}} > 3k\log n.\]

There are $\binom{n}{k} \leq \exp(k\log n)$ choices of $k$ vertices, so the probability that  at least for one choice of adversarial $k$-subset $\sum_{i = k + 1}^n\min(U_i, 2kp\log n) > (n - k)kp(1 - p) + \frac{(n - k)kp(1 - p)}{\gamma(n)}$ is at most

\begin{multline*}
\binom{n}{k}\P{\sum_{i = k + 1}^n\min\left(U_i, 2kp\log n\right) > (n - k)kp(1 - p) + \frac{(n - k)kp(1 - p)}{\gamma(n)}} \leq \\
\leq \binom{n}{k} \cdot 2\exp(-3k\log n) < \exp(-2k\log n).
\end{multline*}
Thus, with probability at least $1 - \exp(-2k\log n)$, for every choice of the $k$-subset we get
$$\sum_{i = k + 1}^n\min\left(U_i, 2kp\log n\right) \leq (1 + o(1))(n - k)kp(1 - p),$$
which finishes the proof of \cref{varbound}.

%% file: main.bbl
\begin{thebibliography}{DGGP14}

\bibitem[AKS98]{AKS98}
Noga Alon, Michael Krivelevich, and Benny Sudakov.
\newblock Finding a large hidden clique in a random graph.
\newblock {\em Random Struct. Algorithms}, 13(3-4):457--466, 1998.

\bibitem[AKV02]{AKV02}
Noga Alon, Michael Krivelevich, and Van~H. Vu.
\newblock On the concentration of eigenvalues of random symmetric matrices.
\newblock {\em Isr. J. Math.}, 131:259--267, 2002.

\bibitem[Ber46]{B46}
S.N. Bernshtein.
\newblock {\em Probability theory (In Russian)}.
\newblock 4 edition, 1946.

\bibitem[DF16]{DF16}
Roee David and Uriel Feige.
\newblock On the effect of randomness on planted 3-coloring models.
\newblock In {\em Proceedings of the 48th Annual {ACM} {SIGACT} Symposium on
  Theory of Computing, {STOC} 2016}, pages 77--90, 2016.

\bibitem[DGGP14]{DGP14}
Yael Dekel, Ori Gurel-Gurevich, and Yuval Peres.
\newblock Finding hidden cliques in linear time with high probability.
\newblock {\em Combinatorics, Probability {\&} Computing}, 23(1):29--49, 2014.

\bibitem[DM15]{DM15}
Yash Deshpande and Andrea Montanari.
\newblock Finding hidden cliques of size $\sqrt{N/e}$ in nearly linear time.
\newblock {\em Foundations of Computational Mathematics}, 15(4):1069--1128,
  2015.

\bibitem[Fei20]{Fei20}
Uriel Feige.
\newblock Introduction to semi-random models.
\newblock In Tim Roughgarden, editor, {\em Beyond the Worst-Case Analysis of
  Algorithms}. 2020.
\newblock to appear.

\bibitem[FK00]{FK00}
Uriel Feige and Robert Krauthgamer.
\newblock Finding and certifying a large hidden clique in a semirandom graph.
\newblock {\em Random Struct. Algorithms}, 16(2):195--208, 2000.

\bibitem[FK03]{FK03}
Uriel Feige and Robert Krauthgamer.
\newblock The probable value of the {L}ov\'asz--{S}chrijver relaxations for
  maximum independent set.
\newblock {\em SIAM J. Comput.}, 32(2):345--370, 2003.

\bibitem[FO08]{FO08}
Uriel Feige and Eran Ofek.
\newblock Finding a maximum independent set in a sparse random graph.
\newblock {\em {SIAM} J. Discrete Math.}, 22(2):693--718, 2008.

\bibitem[FR10]{FR10}
Uriel Feige and Dorit Ron.
\newblock {Finding hidden cliques in linear time}.
\newblock In {\em {21st International Meeting on Probabilistic, Combinatorial,
  and Asymptotic Methods in the Analysis of Algorithms (AofA'10)}}, pages
  189--204, 2010.

\bibitem[HJ12]{HJ12}
Roger~A. Horn and Charles~R. Johnson.
\newblock {\em Matrix Analysis}.
\newblock Cambridge University Press, 2 edition, 2012.

\bibitem[Jer92]{Jerrum92}
Mark Jerrum.
\newblock Large cliques elude the metropolis process.
\newblock {\em Random Struct. Algorithms}, 3(4):347--360, 1992.

\bibitem[Juh82]{Juhasz}
Ferenc Juh{\'a}sz.
\newblock The asymptotic behaviour of {L}ov{\'a}sz’ $\vartheta$ function for
  random graphs.
\newblock {\em Combinatorica}, 2:153--155, 1982.

\bibitem[Kuc95]{Kucera95}
Ludek Kucera.
\newblock Expected complexity of graph partitioning problems.
\newblock {\em Discrete Applied Mathematics}, 57:193--212, 1995.

\bibitem[Lov79]{Lovasz}
L\'aszlo Lov\'asz.
\newblock On the {S}hannon capacity of a graph.
\newblock {\em IEEE Transactions on Information Theory}, 25(1):1--7, 1979.

\bibitem[MPW15]{MPW15}
Raghu Meka, Aaron Potechin, and Avi Wigderson.
\newblock Sum-of-squares lower bounds for planted clique.
\newblock In {\em Proceedings of the Forty-Seventh Annual ACM Symposium on
  Theory of Computing}, STOC ’15, pages 87--96, 2015.

\bibitem[MU04]{MU04}
Kazuhisa Makino and Takeaki Uno.
\newblock New algorithms for enumerating all maximal cliques.
\newblock In {\em SWAT}, pages 260--272, 07 2004.

\bibitem[Vu07]{V07}
Van~H. Vu.
\newblock Spectral norm of random matrices.
\newblock {\em Combinatorica}, 27:721--736, 2007.

\end{thebibliography}
